\documentclass[reqno, eucal]{amsart}

\usepackage[mathscr]{eucal} 

\numberwithin{equation}{section}

\usepackage[dvips]{graphicx}
\usepackage[dvips]{color}
\usepackage{amsmath,amsfonts,amssymb,amsthm,amscd}
\usepackage{epic}
\usepackage{eepic}
\usepackage{longtable}
\usepackage{array}

\newtheorem{theorem}{Theorem}[section]

\newtheorem{proposition}[theorem]{Proposition}

\theoremstyle{definition}

\newtheorem{example}[theorem]{Example}
\newtheorem{remark}[theorem]{Remark}

\newcommand{\Z}{{\mathbb Z}}

\newcommand{\C}{{\mathbb C}}

\newcommand{\ds}{\displaystyle}

\begin{document}

\title[Tetrahedron Equation
and Quantum $R$ Matrices
for Spin Representations]{Tetrahedron Equation
and Quantum $\boldsymbol{R}$ Matrices
for Spin Representations of
$\boldsymbol{B^{(1)}_n}$, $\boldsymbol{D^{(1)}_n}$
 and $\boldsymbol{D^{(2)}_{n+1}}$}

\author{Atsuo Kuniba}
\email{atsuo@gokutan.c.u-tokyo.ac.jp}
\address{Institute of Physics, University of Tokyo, Komaba, Tokyo 153-8902, Japan}

\author{Sergey Sergeev}
\email{Sergey.Sergeev@canberra.edu.au}
\address{Faculty of Information Sciences and Engineering, University of Canberra, Canberra, Australia}

\dedicatory{Dedicated to Professor Vladimir Bazhanov on the occasion of
his sixtieth birthday.}

\maketitle

\vspace{1cm}
\begin{center}{\bf Abstract}
\end{center}
\vspace{0.4cm}
It is known that a solution of the tetrahedron equation generates
infinitely many solutions of the Yang-Baxter equation via suitable reductions.
In this paper this scheme is applied to
an oscillator solution of the tetrahedron equation involving bosons and fermions
by using special 3d boundary conditions.
The resulting solutions of the Yang-Baxter equation are identified
with the quantum $R$ matrices for the spin representations of
$B^{(1)}_n, D^{(1)}_n$ and $D^{(2)}_{n+1}$.

\section{Introduction}\label{sec:intro}
The tetrahedron equation \cite{Zamolodchikov:1980,Zamolodchikov:1981}
is a three-dimensional (3d) extension of
the Yang-Baxter (triangle) equation \cite{Bax}.
It is expressed as an equality between quartic products of
$R$ matrices and/or $L$ operators, and
serves as a sufficient condition for the layer to layer
transfer matrices in
the associated 3d lattice models to commute with each other.
In general, these $R$ matrices and $L$ operators
act on tensor product of three vector spaces
reflecting the three independent directions in the 3d lattice.

Tetrahedron equations possess two notable features.
First, they remain valid under $n$-fold composition of the $L$ operators
in one of the directions for arbitrary $n$.
Namely they straightforwardly generalize to the $n$-layer situation,
which is analogous to the (rather trivial) fact that
a single Yang-Baxter equation implies the
commutativity of row transfer matrices for
arbitrary row lengths in two-dimension (2d).
Second, if one of the three spaces is traced out or
evaluated away appropriately, the tetrahedron equation
reduces to the Yang-Baxter equation among the resulting objects.
We refer to the space so masked
as the ``(third) hidden direction".
In fact, it appears as a space of internal degrees of freedom
attached to each lattice site from the resulting 2d world point of view.

Combining the above two features leads to the following fact:
a solution of the tetrahedron equation
generates an infinite series of solutions of the Yang-Baxter equation.
This phenomenon is known as \emph{dimension-rank} transmutation.
It has been implemented earlier for a certain
3d $L$ operator by taking the trace which corresponds to
the \emph{periodic boundary condition} in the hidden direction \cite{BS06, Bax:1986, BB:1992}.
The resulting solutions of the Yang-Baxter equation have been identified
with the quantum $R$ matrices for a class of finite dimensional
representations of $U_q(\widehat{sl}_n)$.

In this paper
we introduce another type of boundary conditions
in the hidden direction and study the resulting solutions of
the Yang-Baxter equation.
We start from the solution of the tetrahedron equation
consisting of $q$-oscillator 3d $R$ matrix and fermionic 3d $L$ operators,
which are the same as \cite{BS06}.
We construct special  \emph{boundary states}  in a bosonic Fock space
(the hidden direction) and show that they
are eigenvectors of the 3d $R$-matrix, which is the key
to make our reduction scheme work.
By evaluating the $L$ operators with respect to these boundary states,
we derive three series of solutions of the Yang-Baxter equation.
Our main result is that they produce the
quantum $R$ matrices for the spin
representations of
$U_q(B_n^{(1)}), U_q(D_n^{(1)})$
and $U_q(D_{n+1}^{(2)})$
depending on the 3d boundary conditions.
In particular we observe a curious correspondence
between the Dynkin diagrams of these algebras
and the boundary states in the Fock space (Remark \ref{re:dynkin}),
giving a new insight into the quantum group symmetry
of the 3d integrable models.

The layout of the paper is as follows.
In Section \ref{sec:teybe}, we explain general schemes to
obtain series of solutions to the Yang-Baxter equation
from a solution of the tetrahedron equation.
Section \ref{sec:osc} presents a concrete example
of 3d $R$ matrix and 3d $L$ operator in terms of an oscillator algebra
and its Fock representation.
Section \ref{sec:chi} describes the special vectors
in the Fock space which serve as 3d boundary conditions.
We prove the key property that they are eigenvectors of
the 3d $R$ matrix in Proposition \ref{a:pro:chi}.
In Section \ref{a:sec:r3d}, we derive solutions
${\mathscr R}(x)$ of the Yang-Baxter equation from
the 3d $L$ operator via the reduction using the special vectors.
Elements of ${\mathscr R}(x)$ are expressed in an 
matrix product ansatz form.
Section \ref{a:sec:qr} collects formulas for the
spin representations of
$U_q(B_n^{(1)}), U_q(D_n^{(1)})$ \cite{O} and  $U_q(D^{(2)}_{n+1})$
and the associated quantum $R$ matrices $R(x)$
which are necessary for the proof of our main theorem.
Although $U_q(D^{(2)}_{n+1})$ case is just a slight variation of
$U_q(B^{(1)}_n)$, it seems to have been treated nowhere in the
literature so far.
In Section \ref{a:sec:proof}, we present an expository proof of 
our main result ${\mathscr R}(x) = R(x)$ (Theorem \ref{a:th:main}).
Depending on the choices of the special vectors (3d boundary conditions),
the algebras
$U_q(B_n^{(1)}), U_q(D_n^{(1)})$ and  $U_q(D^{(2)}_{n+1})$ are covered.
It suggests a certain correspondence between
the boundary conditions and
relevant Dynkin diagrams (Remark \ref{re:dynkin}).
Our strategy of the proof is to establish that ${\mathscr R}(x)$ satisfies
the standard characterization \cite{Ji} of the
quantum $R$ matrix $R(x)$ \cite{Ji, Baz}, and does not
rely on explicit formulas of the matrix elements.

\section{Tetrahedron equation and Yang-Baxter equation: General scheme}
\label{sec:teybe}
In this section, we explain a general scheme to
generate a series of solutions to the Yang-Baxter equation
from a solution of the tetrahedron equation.

Let ${\mathscr R} _{1,2,3}$ be a linear operator acting on the tensor product of three vector spaces:
\begin{equation}
\mathscr{R}_{1,2,3}\in\mathrm{End}(F\otimes F\otimes F)\;.
\end{equation}
Here the space $F$ (and $V$ coming soon as well)
can be either finite or infinite dimensional
for our general discussion in this section.
We call $\mathscr{R}_{1,2,3}$ (3d) $R$ matrix.
The indices in ${\mathscr R}_{1,2,3}$
are just the reminder of the three copies of $F$ which are labeled
and exhibited when preferable as
\begin{equation}
F\otimes F\otimes F \;=\; \overset{1}{F}\otimes \overset{2}{F}\otimes \overset{3}{F}\;.
\end{equation}
Consider another vector space $V$ and
let ${\mathscr L}_{1,a,b}$ be an operator
acting on the tensor product
$F\otimes V\otimes V$, i.e.,
\begin{equation}\label{a:elop}
{\mathscr L} _{1,a,b}\in\mathrm{End}(F\otimes V\otimes V)\;,
\end{equation}
where again the indices are just
labels (not parameters) of the spaces as
$\overset{1}{F}\otimes \overset{a}{V}\otimes \overset{b}{V}$.
We call ${\mathscr L}_{1,a,b}$ (3d) $L$ operator.
A version of the quantum tetrahedron equation \cite{BS06} is
\begin{equation}\label{TE}
{\mathscr R} _{1,2,3}\,{\mathscr L} _{1,a,b} \,
{\mathscr L} _{2,a,c} {\mathscr L} _{3,b,c} \;=\;
{\mathscr L} _{3,b,c} \,
{\mathscr L} _{2,a,c}\,
 {\mathscr L} _{1,a,b} \,{\mathscr R} _{1,2,3}\;.
\end{equation}
It is an equation in
$\mathrm{End}(\overset{1}{F}\otimes \overset{2}{F} \otimes
\overset{3}{F}\otimes
\overset{a}{V} \otimes \overset{b}{V} \otimes \overset{c}{V})$.
The operators act as identities on the spaces whose labels
are not included in their indices.
The $L$ operators
${\mathscr L} _{1,a,b} , {\mathscr L} _{2,a,c}$ and
${\mathscr L} _{3,b,c} $
are identical except that they act nontrivially on
different sets of tensor components.
The relation (\ref{TE}) can be depicted as Figure \ref{fig:TE}.

\begin{figure}[ht]
\begin{center}
\setlength{\unitlength}{0.25mm}
\begin{picture}(500,200)
\put(0,0)
 {\begin{picture}(200,200) \thinlines
 \drawline[-30](90.00, 42.50)(30.00, 87.50)
 \put(94,32){\scriptsize $a$}
 \drawline[-30](62.50, 37.50)(167.50, 112.50)
 \put(52,30){\scriptsize $b$}
 \drawline[-30](12.50, 75.00)(177.50, 105.00)
 \put(2,70){\scriptsize $c$}
 \Thicklines
 \path(75.00, 25.00)(105.00, 175.00)
 \put(72,15){\scriptsize $1$}
 \path(25.00, 62.50)(115.00, 167.50)
 \put(20,52){\scriptsize $2$}
 \path(162.50, 87.50)(87.50, 162.50)
 \put(165,80){\scriptsize $3$}
 \end{picture}}
\put(300,0)
 {\begin{picture}(200,200) \thinlines
 \drawline[-30](110.00, 157.50)(170.00, 112.50)
 \put(180,100){\scriptsize $a$}
 \drawline[-30](137.50, 162.50)(32.50, 87.50)
 \put(22,75){\scriptsize $b$}
 \drawline[-30](187.50, 125.00)(22.50, 95.00)
 \put(10,90){\scriptsize $c$}
 \Thicklines
 \path(125.00, 175.00)(95.00, 25.00)
 \put(92,5){\scriptsize $1$}
 \path(175.00, 137.50)(85.00, 32.50)
 \put(73,22){\scriptsize $2$}
 \path(37.50, 112.50)(112.50, 37.50)
 \put(120,27){\scriptsize $3$}
 \end{picture}}
\put(245,100){$=$}
\end{picture}
\end{center}
\caption{A pictorial representation of the tetrahedron equation (\ref{TE}).}
\label{fig:TE}
\end{figure}
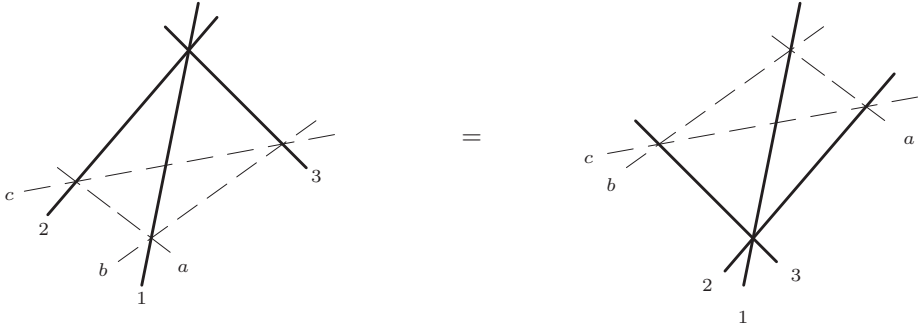

\bigskip
\noindent
We regard this as a one-layer relation.
It is straightforward to generalize it to the $n$-layer case
for any positive integer $n$.
To do so we introduce the $n$-fold tensor product
${\boldsymbol V} = V^{\otimes n}$ and also attach the
labels
$\boldsymbol{a} =(a_1,\ldots, a_n)$ for distinction to these
spaces as
\begin{equation}\label{a:Va}
\overset{\boldsymbol{a}}{\boldsymbol V}= \overset{a_1}{V}
\otimes \overset{a_2}{V}
\otimes \cdots \otimes \overset{a_n}{V}.
\end{equation}
Let $\overset{\boldsymbol{b}}{\boldsymbol V}$
and $\overset{\boldsymbol{c}}{\boldsymbol V}$ be copies
of $\boldsymbol{V}$
with different labels $\boldsymbol{b}$ and $\boldsymbol{c}$.
We compose the elementary $L$ operators (\ref{a:elop})
$n$ times as
\begin{equation}\label{a:cL}
{\mathscr L} _{1,\boldsymbol{a},\boldsymbol{b}}\;=\;
{\mathscr L} _{1,a_1,b_1}{\mathscr L} _{1,a_2,b_2}\cdots
{\mathscr L} _{1,a_n,b_n}\;\in\;
\mathrm{End}(\overset{1}{F}\otimes \overset{\boldsymbol{a}}{\boldsymbol V}
\otimes \overset{\boldsymbol{b}}{\boldsymbol V}).
\end{equation}
Define
${\mathscr L} _{2,\boldsymbol{a},\boldsymbol{c}}$ and
${\mathscr L} _{3,\boldsymbol{b},\boldsymbol{c}}$ similarly.
They act on
$\overset{1}{F}\otimes\overset{2}{F}\otimes\overset{3}{F}\otimes
\overset{\boldsymbol{a}}{\boldsymbol V}
\otimes \overset{\boldsymbol{b}}{\boldsymbol V}
\otimes \overset{\boldsymbol{c}}{\boldsymbol V}$
nontrivially on the components specified by their indices.
Then the elementary tetrahedron equation (\ref{TE})
is lifted up directly to the $n$-layer version:
\begin{equation}\label{TE-2}
{\mathscr R} _{1,2,3}\,
{\mathscr L} _{1,\boldsymbol{a},\boldsymbol{b}}\,
 {\mathscr L} _{2,\boldsymbol{a},\boldsymbol{c}} \,
 {\mathscr L} _{3,\boldsymbol{b},\boldsymbol{c}} \;=\;
{\mathscr L} _{3,\boldsymbol{b},\boldsymbol{c}} \,
{\mathscr L} _{2,\boldsymbol{a},\boldsymbol{c}} \,
{\mathscr L} _{1,\boldsymbol{a},\boldsymbol{b}} \,
{\mathscr R} _{1,2,3}\;.
\end{equation}
In fact, one can carry ${\mathscr R} _{1,2,3}$ through
the $L$ operators by repeated use of (\ref{TE})
from layer to layer.
See Figures \ref{fig:TE2-1}--\ref{fig:TE2-3}.

\begin{figure}[ht]
\scalebox{0.3}{\includegraphics{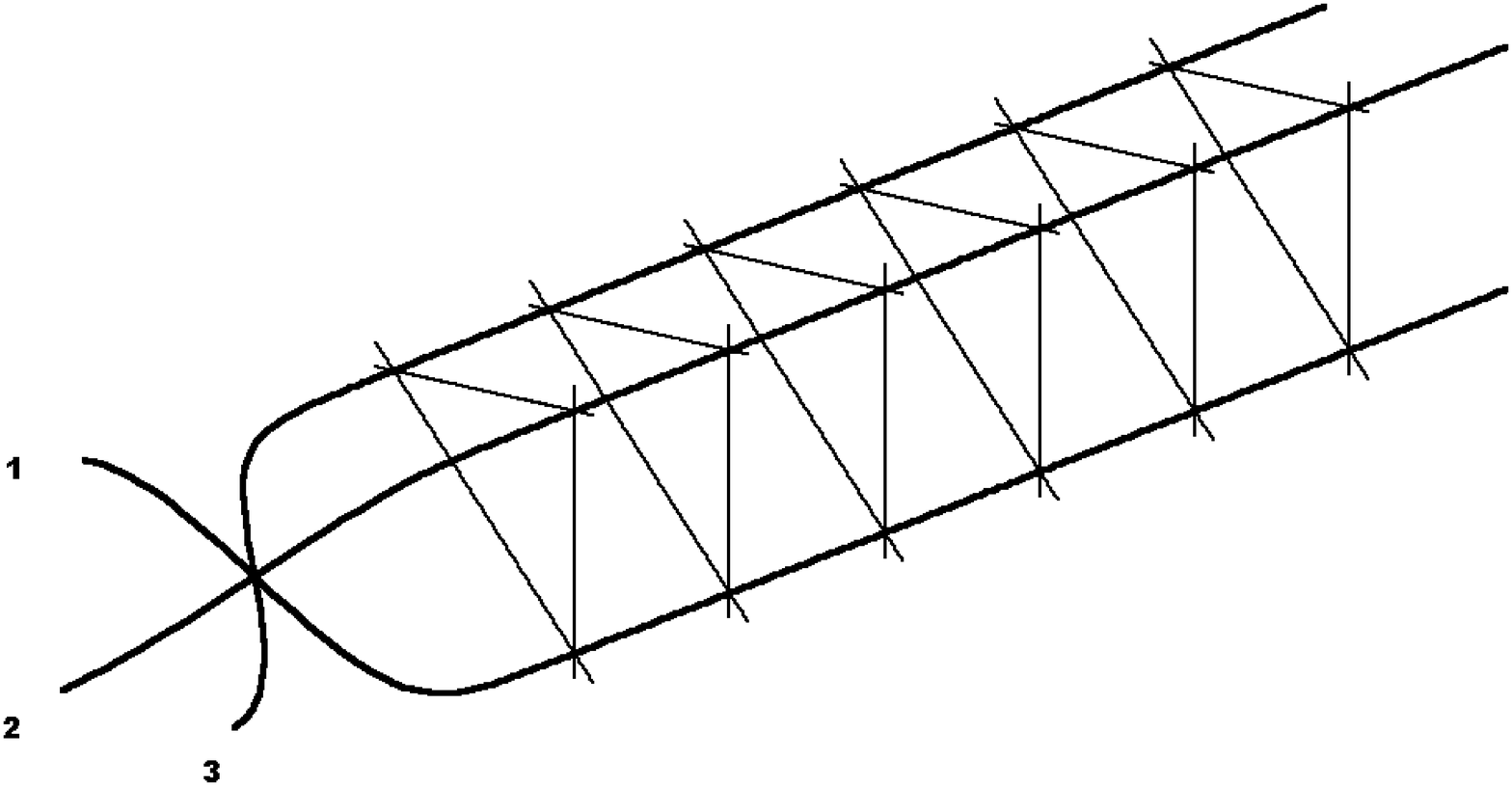}}
\caption{A pictorial representation of the left hand side of the tetrahedron equation (\ref{TE-2}).}
\label{fig:TE2-1}
\end{figure}

\begin{figure}[ht]
\scalebox{0.3}{\includegraphics{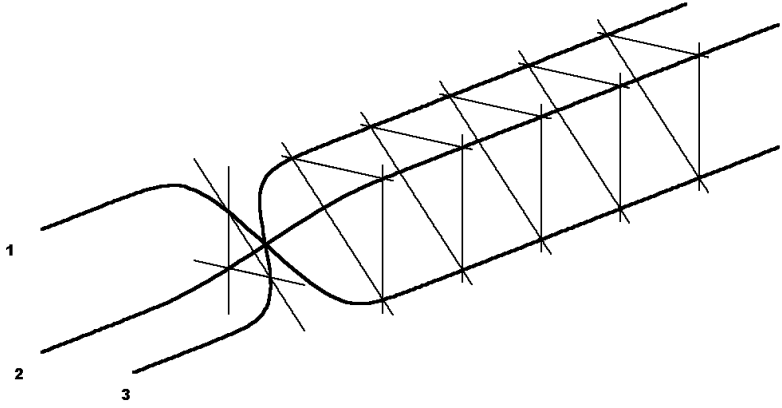}}
\caption{Result of application of the first elementary tetrahedron equation (\ref{TE}) to the left hand side of (\ref{TE-2}).}
\label{fig:TE2-2}
\end{figure}

\begin{figure}[ht]
\scalebox{0.3}{\includegraphics{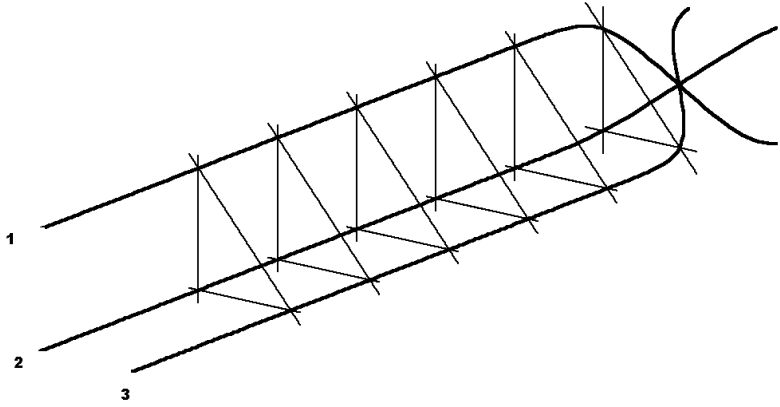}}
\caption{A pictorial representation of the right hand side of the tetrahedron equation (\ref{TE-2}).}
\label{fig:TE2-3}
\end{figure}

\bigskip
Now we explain the prescriptions to reduce
the tetrahedron equation (\ref{TE-2}) to
the Yang-Baxter equation.
The idea is to perform certain evaluations for
$\mathrm{End}(F\otimes F \otimes F)$ part
regarding it as ``internal degrees of freedom" or
a ``hidden direction'' perpendicular to the ``vertices"
corresponding to the other part
$\mathrm{End}({\boldsymbol V}
\otimes {\boldsymbol V} \otimes
{\boldsymbol V})$.
The traditional way to do it is to take the trace \cite{BS06}.
Assuming that ${\mathscr R} _{1,2,3}$ is invertible,
(\ref{TE-2}) leads to the Yang-Baxter equation\footnote{
There is a room to include a
spectral parameter by inserting a ``diagonal field" in taking the trace.
See \cite{BS06} for detail.}
\begin{equation}
{\mathscr R}_{{\boldsymbol a},{\boldsymbol b}}\,
{\mathscr R}_{{\boldsymbol a},{\boldsymbol c}}\,
{\mathscr R}_{{\boldsymbol b},{\boldsymbol c}}
=
{\mathscr R}_{{\boldsymbol b},{\boldsymbol c}}\,
{\mathscr R}_{{\boldsymbol a},{\boldsymbol c}}\,
{\mathscr R}_{{\boldsymbol a},{\boldsymbol b}} \in
\mathrm{End}(\overset{\boldsymbol a}{\boldsymbol V}
\otimes \overset{\boldsymbol b}{\boldsymbol V} \otimes \overset{\boldsymbol c}{\boldsymbol V})
\end{equation}
for the $R$ matrix defined by
\begin{equation}
{\mathscr R}_{{\boldsymbol a},{\boldsymbol b}}
= \mathrm{Tr}_{F}
({\mathscr L} _{1,\boldsymbol{a},\boldsymbol{b}})
\in \mathrm{End}
(\overset{\boldsymbol a}{\boldsymbol V}\otimes \overset{\boldsymbol b}{\boldsymbol V}),
\end{equation}
where $F$ here is actually the first copy $\overset{1}{F}$.
The other ones
${\mathscr R}_{{\boldsymbol b},{\boldsymbol c}}$ and
${\mathscr R}_{{\boldsymbol a},{\boldsymbol c}}$
are defined similarly  and they are identical
except the nontrivially acting components in
$\overset{\boldsymbol a}{\boldsymbol V}
\otimes \overset{\boldsymbol b}{\boldsymbol V}
\otimes \overset{\boldsymbol c}{\boldsymbol V}$.
For a concrete result along this line, we refer to \cite{BS06},
which has reproduced a class of quantum $R$ matrices for
$U_q(\widehat{sl}_{n})$.

In this paper we consider a new scenario.
Namely, suppose there are vectors
\begin{equation}
|\chi_s(x,y)\rangle
=|\chi_s(x)\rangle \otimes |\chi_s(xy)\rangle
\otimes |\chi_s(y)\rangle
\in F\otimes F\otimes F\;,
\end{equation}
where $x,y$ are extra (spectral) parameters, such that
\begin{equation}\label{RX}
{\mathscr R} _{1,2,3} |\chi_s(x,y)\rangle   = |\chi_s(x,y)\rangle \;.
\end{equation}
The index $s$ is a label of (possibly more than one) such vectors.
Suppose also similar vectors exist in the dual space:
\begin{equation}
\langle \overline \chi_s(x,y)|
=\langle \overline\chi_s(x)| \otimes
\langle \overline\chi_s(xy)| \otimes
\langle \overline\chi_s(y)|
\in F^*\otimes F^*\otimes F^*\;,
\end{equation}
with the property
\begin{equation}\label{XR}
\langle \overline \chi_s(x,y) |
{\mathscr R} _{1,2,3}   =\langle \overline \chi_s(x,y) |\;.
\end{equation}
Then, evaluating the tetrahedron equation (\ref{TE-2}) between
$\langle \overline \chi_s(x,y)|$ and
$|\chi_t(1,1)\rangle $\footnote{In general, $|\chi_t(x',y')\rangle $
can be used from the right. However, in our examples
treated later, such a freedom is absorbed elsewhere
and becomes equivalent to
$|\chi_t(1,1)\rangle$.},
one produces the Yang-Baxter equation
\begin{equation}\label{YBE}
{\mathscr R}_{\boldsymbol{a},\boldsymbol{b}}(x)
{\mathscr R}_{\boldsymbol{a},\boldsymbol{c}}(xy)
{\mathscr R}_{\boldsymbol{b},\boldsymbol{c}}(y)
=
{\mathscr R}_{\boldsymbol{b},\boldsymbol{c}}(y)
{\mathscr R}_{\boldsymbol{a},\boldsymbol{c}}(xy)
{\mathscr R}_{\boldsymbol{a},\boldsymbol{b}}(x)\;
\in
\mathrm{End}(\overset{\boldsymbol a}{\boldsymbol V}
\otimes \overset{\boldsymbol b}{\boldsymbol V} \otimes \overset{\boldsymbol c}{\boldsymbol V}),
\end{equation}
which forms a series corresponding to the choices $n=1,2,\ldots$.
The $R$ matrices here are obtained from the $L$ operator
by the dual pairing of
$\overset{1}{F}{}^\ast$ and $\overset{1}{F}$ as\footnote{
We regard
$\mathrm{End}(\overset{\boldsymbol a}{\boldsymbol V}
\otimes \overset{\boldsymbol b}{\boldsymbol V})$ here as
naturally embedded into
$\mathrm{End}(\overset{\boldsymbol a}{\boldsymbol V}
\otimes \overset{\boldsymbol b}{\boldsymbol V} \otimes \overset{\boldsymbol c}{\boldsymbol V})$ in (\ref{YBE}).}

\begin{equation}\label{a:xLx}
{\mathscr R}_{\boldsymbol{a},\boldsymbol{b}}(x)\,
(={\mathscr R}^{s, t}_{\boldsymbol{a},\boldsymbol{b}}(x) ) =
\langle\overline\chi_s(x)|
 {\mathscr L} _{1,\boldsymbol{a},\boldsymbol{b}}
 |\chi_t(1)\rangle
 \in\textrm{End}(
 \overset{\boldsymbol a}{\boldsymbol V}
\otimes \overset{\boldsymbol b}{\boldsymbol V})\;\; \textrm{etc}.
\end{equation}
Note the extra option regarding the choices of $s$ and $t$.
In fact, our main Theorem \ref{a:th:main}
will utilize this degree of freedom to
cover the three affine Lie algebras
$B^{(1)}_n,  D^{(1)}_n$ and $D^{(2)}_{n+1}$ in a
unified scheme.
Exhibiting the dependence on $s, t$ (and suppressing the
trivial reference to the labels ${\boldsymbol a}, {\boldsymbol b}$ of the
tensor components),
we will write the $R$ matrix
also as ${\mathscr R}^{s, t}(x) \in
\mathrm{End}(\boldsymbol{V} \otimes \boldsymbol{V})$.
One may view the bra and ket vectors in (\ref{a:xLx}) as specifying the
special boundary condition along the hidden direction as
in Figure \ref{fig:Rst}. See also Remark \ref{re:dynkin}.

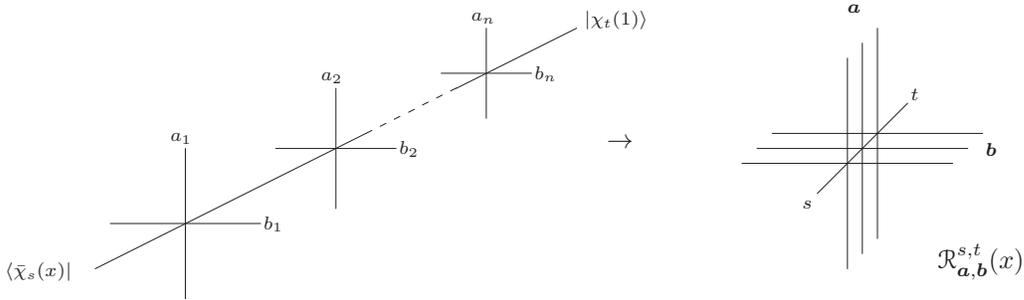
\begin{figure}[ht]
\setlength{\unitlength}{0.20mm}
\begin{picture}(650,200)
\put(0,0){\begin{picture}(400,200)
%
\path(40,20)(220,110)
\path(280,140)(360,180)
\dashline[0]{5}(220,110)(280,140)
\path(100,0)(100,100)\path(50,50)(150,50)
\path(200,60)(200,140)\path(160,100)(240,100)
\path(300,120)(300,180)\path(270,150)(330,150)
\put(-20,15){\scriptsize $\langle\bar{\chi}_s(x)|$}
\put(365,182){\scriptsize $|\chi_t(1)\rangle$} \put(90,105){\scriptsize
$a_1$}\put(190,145){\scriptsize $a_2$}
\put(290,185){\scriptsize $a_n$}
\put(152,46){\scriptsize $b_1$} \put(242,96){\scriptsize
$b_2$} \put(332,146){\scriptsize $b_n$}
\end{picture}}
\put(380,100){$\to$}
\put(450,0){\begin{picture}(200,200)
\path(70,70)(130,130)
\put(60,60){\scriptsize $s$}
\put(132,132){\scriptsize $t$}
\path(90,20)(90,160)\path(100,30)(100,170)\path(110,40)(110,180)
\path(20,90)(160,90)\path(30,100)(170,100)\path(40,110)(180,110)
\put(182,95){\scriptsize
$\boldsymbol{b}$} \put(90,190){\scriptsize $\boldsymbol{a}$}
\put(150,20){${\mathscr R}^{s, t}_{\boldsymbol{a},\boldsymbol{b}}(x)$}
\end{picture}}
\end{picture}
\caption{A pictorial representation of
${\mathscr R}^{s, t}_{\boldsymbol{a},\boldsymbol{b}}(x)$ (\ref{a:xLx}).}
\label{fig:Rst}
\end{figure}

There is another version of the $R$ matrix
\begin{equation}\label{a:Rc}
\check{\mathscr R}(x) \,(= \check{\mathscr R}^{s, t}(x))
= P \,{\mathscr R}^{s,t}(x)\qquad
(P(u\otimes v) = v \otimes u),
\end{equation}
in terms of which the Yang-Baxter equation takes another familiar form:
\begin{equation}
\label{a:ybe}
(\check{\mathscr R}(x) \otimes 1)(1\otimes \check{\mathscr R}(xy))(1\otimes \check{\mathscr R}(y))
=(1\otimes \check{\mathscr R}(y))(\check{\mathscr R}(xy) \otimes 1)(1\otimes \check{\mathscr R}(x)).
\end{equation}
We will use the both versions ${\mathscr R}^{s, t}(x)$
and $\check{\mathscr R}^{s, t}(x)$ for convenience.

The rest of the paper is devoted to a concrete realization of the above
scheme with the following choice:
\begin{equation}\label{a:choice}
\begin{split}
&\boldsymbol{V} = V^{\otimes n},\quad
V = \C^2\; \;(\text{fermionic Fock space}),\\
&F = \bigoplus_{m \in \Z_{\ge 0}} \C|m\rangle\;\;
(\text{bosonic Fock space}).
\end{split}
\end{equation}

\section{Oscillators and the tetrahedron equation}\label{sec:osc}
Let us present  an example of the
$R$ matrix and $L$ operators satisfying the
tetrahedron equation (\ref{TE}).
Let ${\mathcal A}$ be
the associative algebra (called oscillator algebra)
generated by
$\rm{\bf  a}^+, {\rm {\bf a}}^-, {\rm {\bf k}}$ with the relations
\begin{align}\label{a:kaa}
{\rm {\bf k}} \,{\rm {\bf a}}^{\pm} = p^{\pm 1}{\rm {\bf a}}^{\pm}\,{\rm {\bf k}},\quad {\rm {\bf a}}^+ {\rm {\bf a}}^- =1-p^{-1}{\rm {\bf k}}^2,\quad
{\rm {\bf a}}^- {\rm {\bf a}}^+ = 1-p\, {\rm {\bf k}}^2.
\end{align}
Here $p$ is an indeterminate.
We use the representation on the bosonic Fock space
$F=\bigoplus_{m \in \Z_{\ge 0}} \C|m\rangle$ as follows:
\begin{equation}\label{a:frep}
\begin{split}
&{\rm {\bf a}}^+|m\rangle =\sqrt{1-p^{2m+2}}|m+1\rangle,\quad
{\rm {\bf a}}^-|m\rangle =\sqrt{1-p^{2m}}|m-1\rangle,\quad
{\rm {\bf k}} | m \rangle = p^{m+\frac{1}{2}} | m \rangle,\\
&\langle m | {\rm {\bf a}}^- =\langle m+1 |\sqrt{1-p^{2m+2}},\quad
\langle m | {\rm {\bf a}}^+ =\langle m-1 |\sqrt{1-p^{2m}},\quad
\langle m | {\rm {\bf k}}= \langle m | p^{m+\frac{1}{2}}.
\end{split}
\end{equation}
The vector $|0\rangle$ is the total vacuum.
Let $V=\C^2$ as in (\ref{a:choice}) and
introduce an $L$ operator  ${\mathscr L} \in
\mathrm{End}(F \otimes V\otimes V)$  by
\begin{equation}\label{a:Lmat}
{\mathscr L} = \bigl({\mathscr L} (\alpha', \beta' | \alpha, \;\beta)
\bigr)_{(\alpha',\beta'), (\alpha,\beta)},\quad
{\mathscr L} (\alpha', \beta' | \alpha, \;\beta)
=\begin{pmatrix}
1 & 0 & 0 & 0\\
0 & -i{\rm {\bf k}} & {\rm {\bf a}}^+ & 0\\
0 & {\rm {\bf a}}^- & -i{\rm {\bf k}} & 0\\
0 & 0 & 0 & 1
\end{pmatrix},
\end{equation}
where ${\rm {\bf a}}^\pm, {\rm {\bf k}}$ are regarded as
representations (\ref{a:frep}).
The row index $(\alpha',\beta')$ and the column index
$(\alpha, \beta)$  are arranged in the order
$(0,0), (1,0), (0,1), (1,1)$
form top to bottom and left to right, respectively.
The $0$ and $1$ label the base vectors of $V$ corresponding to
the fermionic states.
In this interpretation,  the relations (\ref{a:kaa}) are in fact free-fermion conditions for ${\mathscr L} $ \cite{Bax,S09a}.
The $L$ operator  (\ref{a:Lmat}) is traditionally written also as
${\mathscr L}[{\mathcal A}]$.

Let  ${\mathcal A}_i$ be the copy of ${\mathcal A}$
with generators
${\rm {\bf a}}^\pm_i, {\rm {\bf k}}_i\, (i=1,2,3)$ that
act on the $i$th component of
$\overset{1}{F}\otimes \overset{2}{F} \otimes \overset{3}{F}$.
Now one can consider the $L$ operators
\begin{equation}
{\mathscr L} _{a,b}[\mathcal{A}_1],\;
{\mathscr L} _{a,c}[\mathcal{A}_2],\;
{\mathscr L} _{b,c}[\mathcal{A}_3]
\in \mathrm{End}(\overset{1}{F}\otimes \overset{2}{F} \otimes \overset{3}{F}
\otimes \overset{a}{V}\otimes \overset{b}{V}\otimes\overset{c}{V})
\end{equation}
acting nontrivially on the slots specified by their indices.
\begin{theorem}\label{a:th:RLLL}
There is a
unique (up to a constant multiple) invertible operator
${\mathscr R}\in\mathrm{End}(F\otimes F\otimes F)$ such that
\begin{equation}\label{theTE}
{\mathscr R} \;
{\mathscr L} _{a,b}[\mathcal{A}_1]
{\mathscr L} _{a,c}[\mathcal{A}_2]
{\mathscr L} _{b,c}[\mathcal{A}_3]
=
{\mathscr L} _{b,c}[\mathcal{A}_3]
{\mathscr L} _{a,c}[\mathcal{A}_2]
{\mathscr L} _{a,b}[\mathcal{A}_1]
 \; {\mathscr R}.
\end{equation}
\end{theorem}
\begin{proof}
A proof of this theorem can be found in \cite{BS06,BMS:2008}.
The matrix equation (\ref{theTE}) can be solved
straightforwardly in the form\footnote{We omit a formula for
$ {\mathscr R}  \, ({\rm {\bf k}}_2)^2 \,{\mathscr R}^{-1}$ as it will not be used
in this paper. See \cite{BS06}.}
\begin{equation}\label{themap}
\left\{\begin{array}{lll}
\ds {\mathscr R}  \;\; {\rm {\bf k}}_2^{}{\rm{\bf a}}_1^{\pm} \;{\mathscr R}^{-1} \;=\; {\rm {\bf k}}_3^{} {\rm{\bf a}}_1^{\pm} + {\rm {\bf k}}_1^{} {\rm{\bf a}}_2^{\pm} {\rm{\bf a}}_3^{\mp}\;,\\
[2mm]
\ds  {\mathscr R} \;\;  {\rm{\bf a}}_2^{\pm}  \;{\mathscr R}^{-1} \;=\; {\rm{\bf a}}_1^{\pm} {\rm{\bf a}}_3^{\pm} - {\rm {\bf k}}_1^{} {\rm {\bf k}}_3^{} {\rm{\bf a}}_2^{\pm}\;,\\
[2mm]
\ds {\mathscr R}  \;\;  {\rm {\bf k}}_2^{} {\rm{\bf a}}_3^{\pm}  \;  {\mathscr R}^{-1} \;=\; {\rm {\bf k}}_1^{} {\rm{\bf a}}_3^{\pm} + {\rm {\bf k}}_3^{} {\rm{\bf a}}_1^{\mp} {\rm{\bf a}}_2^{\pm}\;,\\
[2mm]
{\mathscr R} \;\;  {\rm {\bf k}}_1 {\rm {\bf k}}_2 \, {\mathscr R}^{-1} = {\rm {\bf k}}_1 {\rm {\bf k}}_2,\\
[2mm]
{\mathscr R} \;\;  {\rm {\bf k}}_2 {\rm {\bf k}}_3 \, {\mathscr R}^{-1} = {\rm {\bf k}}_2 {\rm {\bf k}}_3.
\end{array}\right.
\end{equation}
One can easily check that (\ref{themap}) defines the automorphism of $\mathcal{A}^{\otimes 3}$. In addition, the Fock space representation is irreducible. Therefore, ${\mathscr R}$ exists and is unique up to a constant multiple.
\end{proof}

\noindent
The relation (\ref{theTE}) is equivalent to quantum Korepanov
equation, it can be also seen as the
tetrahedral Zamolodchikov algebra/local Yang-Baxter equation for the adjoint action of ${\mathscr R}$. See the long story of \cite{MN:1989,Korepanov:1993jsp,Korepanov:1995,KKS:1998,BS06,BMS:2008} for details.
We fix the normalization of ${\mathscr R}$ by
\begin{equation}
{\mathscr R}|0\rangle = |0\rangle\;.
\end{equation}
The relation (\ref{theTE}) evidently possesses the tetrahedral structure (\ref{TE})
by the identification
\begin{equation}
{\mathscr L} _{a,b}[\mathcal{A}_1] = {\mathscr L} _{1,a,b}\;,\;
{\mathscr L} _{a,c}[\mathcal{A}_2] = {\mathscr L} _{2,a,c}\;,\;
{\mathscr L} _{b,c}[\mathcal{A}_3] = {\mathscr L} _{3,b,c}\;,
\quad {\mathscr R}={\mathscr {\mathscr R}}_{1,2,3}.
\end{equation}
One can verify as well, the adjoint action (\ref{themap}) of ${\mathscr R}$ coincides with the inverse adjoint action, therefore
\begin{equation}\label{a:Rinv}
{\mathscr R}^{-1}={\mathscr R}\;.
\end{equation}
An explicit formula for the matrix elements of ${\mathscr R}$
is included in Appendix \ref{sec:app}.
Although, we will only need the relations
(\ref{themap}) and (\ref{a:Rinv})  later in this paper.

\section{\mathversion{bold} Special vectors in Fock space}\label{sec:chi}
Let us give two vectors
$|\chi_1(x,y)\rangle$ and $|\chi_2(x,y)\rangle$ 
in the Fock space $F\otimes F \otimes F$ (and their duals)
having the properties (\ref{RX}) and (\ref{XR}),
which are the key to our construction. 
First we introduce the following vectors in $F$ and $F^\ast$:
\begin{align}
&| \chi_1(x)\rangle =
\frac{1}{(x\, {\rm {\bf a}}^+; p)_\infty}|0 \rangle,\qquad
| \chi_2(x)\rangle =
\frac{1}{(x({\rm {\bf a}}^+)^2; p^4)_\infty}|0 \rangle,\label{a:chi1}\\
&\langle \overline{\chi}_1(x) |
= \langle 0 |\frac{1}{(x \, {\rm {\bf a}}^-; p)_\infty},\qquad
\langle \overline{\chi}_2(x) |
= \langle 0 |\frac{1}{(x ({\rm {\bf a}}^-)^2; p^4)_\infty},
\label{a:chi2}
\end{align}
where $(x; p)_j = \prod_{i=1}^j(1-xp^{i-1})$ as usual. 
These vectors were introduced in \cite{S09b} without a proof of their properties.
It is straightforward to show
\begin{alignat}{2}
&\Bigl({\rm {\bf a}}^+-x^{-1}(1-p^{-\frac{1}{2}}{\rm {\bf k}})\Bigr)|\chi_1(x)\rangle = 0,\qquad&
&\Bigl({\rm{\bf a}}^--x(1+p^{\frac{1}{2}}{\rm {\bf k}})\Bigr)|\chi_1(x)\rangle = 0,\label{a:aket1}\\
&\langle \overline{\chi}_1(x) |
\Bigl({\rm {\bf a}}^--x^{-1}(1-p^{-\frac{1}{2}}{\rm {\bf k}})\Bigr)=0,&\quad
&\langle \overline{\chi}_1(x) |
\Bigl({\rm {\bf a}}^+-x(1+p^{\frac{1}{2}}{\rm {\bf k}})\Bigr)=0,\label{a:abra1}\\
&({\rm {\bf a}}^--x{\rm {\bf a}}^+)|\chi_2(x)\rangle =0,&
&\langle \overline{\chi}_2(x) | ({\rm {\bf a}}^+-x{\rm {\bf a}}^-)=0.\label{a:aket2}
\end{alignat}
By eliminating ${\rm {\bf k}}$ in (\ref{a:aket1}) and (\ref{a:abra1}), we also have
\begin{align}
&(x^{-1}{\rm {\bf a}}^-+px{\rm {\bf a}}^+-1-p)|\chi_1(x)\rangle = 0,\label{a:aket3}\\
&\langle \overline{\chi}_1(x) | (x^{-1}{\rm {\bf a}}^++px{\rm {\bf a}}^--1-p)=0.
\end{align}
Conversely any one  of the three equations in (\ref{a:aket1})  and
(\ref{a:aket3})  serve as
characterization of $|\chi_1(x)\rangle$ up to an overall normalization.
Similarly the first equation in (\ref{a:aket2}) fixes $|\chi_2(x)\rangle$
up to an overall scalar.
With regard to the dual vectors, the situation is parallel.
Now the vectors $|\chi_s(x,y)\rangle\,(s=1,2)$ and their duals are
defined by
\begin{equation}\label{a:kai}
\begin{split}
|\chi_s(x,y)\rangle
&= |\chi_s(x)\rangle \otimes |\chi_s(xy)\rangle \otimes |\chi_s(y)\rangle
\in F \otimes F \otimes F,\\
\langle\overline{\chi}_s(x,y)| &=
\langle \overline{\chi}_s(x) | \otimes
\langle \overline{\chi}_s(xy) | \otimes
\langle \overline{\chi}_s(y) |
\in F^* \otimes F^* \otimes F^*.
\end{split}
\end{equation}
Thus by setting $(u_1,u_2,u_3)=(x,xy,y)$, the vector
$|\chi_1(x,y)\rangle$
is characterized by
\begin{align}
&(u_i{\rm {\bf a}}^+_i-1+p^{-\frac{1}{2}}{\rm {\bf k}}_i)|\chi_1(x,y)\rangle = 0\;\;\text{or}
\label{a:ch11}\\
&(u_i^{-1}{\rm {\bf a}}^-_i-1-p^{\frac{1}{2}}{\rm {\bf k}}_i)|\chi_1(x,y)\rangle = 0\;\;\text{or}
\label{a:ch12}\\
&(u_i^{-1}{\rm {\bf a}}^-_i+pu_i{\rm {\bf a}}^+_i-1-p)|\chi_1(x,y)\rangle = 0
\label{a:ch13}
\end{align}
for each $i=1,2,3$,
and so is $|\chi_2(x,y)\rangle$ by
\begin{equation}
({\rm {\bf a}}^-_i- u_i{\rm {\bf a}}^+_i)|\chi_2(x,y)\rangle = 0 \quad
(i=1,2,3)\label{a:ch2}
\end{equation}
in addition to the normalization
$|\chi_s(x,y)\rangle
= |0\rangle \otimes  |0\rangle \otimes  |0\rangle +
\text{non-vacuum terms}$.
Similar characterization holds also for
$\langle\overline{\chi}_s(x,y)|$.
In what follows we denote ${\mathscr R}_{1,2,3}$
simply by ${\mathscr R}$.
\begin{proposition}\label{a:pro:chi}
The vector $|\chi_s(x,y)\rangle $ and its dual
satisfy the relations (\ref{RX}) and (\ref{XR}).
Namely, the following equalities are valid for $s=1,2$:
\begin{equation}
{\mathscr R}|\chi_s(x,y)\rangle  = |\chi_s(x,y)\rangle ,\quad
\langle\overline{\chi}_s(x,y)|{\mathscr R} = \langle\overline{\chi}_s(x,y)|.
\end{equation}
\end{proposition}
\begin{proof}
We shall only treat the first relation. The second one is similarly
derived.
It is easy to see that
${\mathscr R}|\chi_s(x,y)\rangle$ satisfies the same normalization
condition as $|\chi_s(x,y)\rangle$ mentioned after
(\ref{a:ch2}).

\smallskip
$\bullet$ {\it Proof of
${\mathscr R}|\chi_2(x,y)\rangle  = |\chi_2(x,y)\rangle$}.
It then suffices to check
that ${\mathscr R}|\chi_2(x,y)\rangle$ also
fulfills (\ref{a:ch2}).
The $i=1$ case is shown by multiplying the invertible
element ${\mathscr R}\,{\rm {\bf k}}_2$ as
\begin{align*}
&{\mathscr R}\,{\rm {\bf k}}_2({\rm {\bf a}}^-_1-x{\rm {\bf a}}^+_1){\mathscr R}|\chi_2(x,y)\rangle \\
&\overset{(\ref{themap})}{=}
\bigl(({\rm {\bf k}}_3{\rm {\bf a}}^-_1+{\rm {\bf k}}_3{\rm {\bf a}}^-_2{\rm {\bf a}}^+_3)
-x({\rm {\bf k}}_3{\rm {\bf a}}^+_1+{\rm {\bf k}}_3{\rm {\bf a}}^+_2{\rm {\bf a}}^-_3)\bigr)|\chi_2(x,y)\rangle \\
&\overset{\phantom{(\ref{themap})}}{=}
{\rm {\bf k}}_3({\rm {\bf a}}^-_1-x{\rm {\bf a}}^+_1)|\chi_2(x,y)\rangle
+ {\rm {\bf k}}_3({\rm {\bf a}}^-_2{\rm {\bf a}}^+_3-x{\rm {\bf a}}^+_2{\rm {\bf a}}^-_3)|\chi_2(x,y)\rangle
\overset{(\ref{a:ch2})}{=} 0.
\end{align*}
The cases $i=2,3$ in (\ref{a:ch2}) can be checked
in the same manner.

\smallskip
$\bullet$
{\it Proof of ${\mathscr R}|\chi_1(x,y)\rangle  = |\chi_1(x,y)\rangle$}.
First we show $(\ref{a:ch13})_{i=2}$, i.e.,
\begin{equation*}
\bigl((xy)^{-1}{\rm {\bf a}}^-_2+pxy{\rm {\bf a}}^+_2-1-p\bigr)
{\mathscr R}|\chi_1(x,y)\rangle =0.
\end{equation*}
By multiplying ${\mathscr R}$ and applying
(\ref{themap}), the LHS becomes
\begin{align*}
\bigl((xy)^{-1}({\rm {\bf a}}^-_1{\rm {\bf a}}^-_3-{\rm {\bf k}}_1{\rm {\bf k}}_3{\rm {\bf a}}^-_2)+pxy
({\rm {\bf a}}^+_1{\rm {\bf a}}^+_3-{\rm {\bf k}}_1{\rm {\bf k}}_3{\rm {\bf a}}^+_2)-1-p\bigr) |\chi_1(x,y)\rangle.
\end{align*}
Eliminate ${\rm {\bf a}}^-_i$ by (\ref{a:ch13}).
The result reads
\begin{equation*}
(1+p)\bigl(p(1-x{\rm {\bf a}}^+_1)(1-y{\rm {\bf a}}^+_3)-{\rm {\bf k}}_1{\rm {\bf k}}_3\bigr)
|\chi_1(x,y)\rangle.
\end{equation*}
This indeed vanishes due to (\ref{a:ch11}).
It follows that another characterizing property
$(\ref{a:ch11})_{i=2}$, i.e.,
$\bigl(xy{\rm {\bf a}}^+_2-(1-p^{-\frac{1}{2}}{\rm {\bf k}}_2)\bigr)
{\mathscr R}|\chi_1(x,y)\rangle=0$ also holds.
Multiplying it with ${\mathscr R}$ and using (\ref{themap}) again,
we get
$\bigl(xy({\rm {\bf a}}^+_1{\rm {\bf a}}^+_3-
{\rm {\bf k}}_1{\rm {\bf k}}_3{\rm {\bf a}}^+_2)-1
+p^{-\frac{1}{2}}{\rm {\bf k}}'_2\bigr)
|\chi_1(x,y)\rangle=0$,
where ${\rm {\bf k}}'_2={\mathscr R}\, {\rm {\bf k}}_2 {\mathscr R}$.
Eliminating ${\rm {\bf a}}^+_i$ by (\ref{a:ch11}),
this is rewritten as
\begin{equation}\label{a:kkk}
p^{-\frac{1}{2}}\bigl(-{\rm {\bf k}}_1
-{\rm {\bf k}}_3
-(p^{\frac{1}{2}} - p^{-\frac{1}{2}}){\rm {\bf k}}_1{\rm {\bf k}}_3
+{\rm {\bf k}}_1{\rm {\bf k}}_2{\rm {\bf k}}_3+{\rm {\bf k}}'_2)|\chi_1(x,y)\rangle=0.
\end{equation}
Now we can derive the remaining
relations $(\ref{a:ch11})_{i=1,3}$ for ${\mathscr R}|\chi_1(x,y)\rangle$.
For example in the $i=1$ case,
multiplication of ${\mathscr R} \,{\rm {\bf k}}_2$ and application of (\ref{themap})
amount to showing
\begin{equation*}
\bigl(x({\rm {\bf k}}_3{\rm {\bf a}}^+_1
+{\rm {\bf k}}_1{\rm {\bf a}}^+_2{\rm {\bf a}}^-_3)
-{\rm {\bf k}}'_2+p^{-\frac{1}{2}}{\rm {\bf k}}_1{\rm {\bf k}}_2
\bigr)|\chi_1(x,y)\rangle=0.
\end{equation*}
Rewriting ${\rm {\bf a}}^{\pm}_i$ in terms of ${\rm {\bf k}}_i$
by (\ref{a:ch11}) and (\ref{a:ch12}),
one finds the resulting vector is proportional
to the LHS of (\ref{a:kkk}) hence zero.
The equality $(\ref{a:ch11})_{i=3}$ can be confirmed in the same way.
\end{proof}

\section{2d reduction of 3d $L$ operator}\label{a:sec:r3d}
We are ready to construct $R$ matrices satisfying the Yang-Baxter equation
following the prescription (\ref{a:xLx}).
Note that the algebra ${\mathcal A}$ is naturally decomposed into
the direct sum
\begin{equation}\label{a:deco}
{\mathcal A} =
{\mathcal A} _{++}\oplus
{\mathcal A} _{+-}\oplus
{\mathcal A} _{-+}\oplus
{\mathcal A} _{--},
\end{equation}
where ${\mathcal A}_{\varepsilon_1, \varepsilon_2}$
is the joint eigenspace of the involutive
automorphism $\sigma_1, \sigma_2$ of  ${\mathcal A}$:
\begin{align}
&{\mathcal A}_{\varepsilon_1, \varepsilon_2}
=\{ x \in {\mathcal A}\mid \sigma_i(x) = \varepsilon_i x\; (i=1,2)\},
\label{a:aee}\\
&\sigma_1({\rm {\bf a}}^\pm)=-{\rm {\bf a}}^\pm, \;
\sigma_1({\rm {\bf k}})={\rm {\bf k}},\quad
\sigma_2({\rm {\bf a}}^\pm)={\rm {\bf a}}^\pm, \;
\sigma_2({\rm {\bf k}})=-{\rm {\bf k}}.
\label{a:sig12}
\end{align}
Using the vectors (\ref{a:chi1}) and (\ref{a:chi2}),
we introduce the linear forms
$\langle \phantom{O} \rangle_{11},\,
\langle \phantom{O} \rangle_{21}$ and
$\langle \phantom{O} \rangle_{22}$
on ${\mathcal A}$ by
\begin{align}
&\langle {\mathcal O}\rangle_{s1}
=\frac{\langle \overline{\chi}_s(x) |
{\mathcal O}|\chi_1(1)\rangle}{\langle \overline{\chi}_s(x)|\chi_1(1)\rangle}
\quad (s=1,2,\,
{\mathcal O} \in {\mathcal A}),
\label{a:bracket}\\
&\langle {\mathcal O}\rangle_{2 2}
=\frac{\langle \overline{\chi}_2(x) |
{\mathcal O}|\chi_2(1)\rangle}
{\langle \overline{\chi}_2(x)|(-i{\rm {\bf k}})^{(1\mp1)/2}|\chi_2(1)\rangle}
\quad (
{\mathcal O} \in {\mathcal A}_{+\pm}\oplus {\mathcal A}_{-\pm}),
\label{a:bracket2}
\end{align}
where
$\langle {\mathcal O}\rangle_{2 2} = 0$
for ${\mathcal O} \in {\mathcal A}_{-\pm}$.
The denominators are simple factors as
\begin{align}
&\langle \overline{\chi}_s(x)|\chi_1(1)\rangle
= \frac{(-px; p^s)_\infty}{(x; p^s)_\infty}\langle 0 | 0 \rangle\quad (s=1,2),
\label{a:norm}\\
&\langle \overline{\chi}_2(x)|\chi_2(1)\rangle =
ip^{-\frac{1}{2}}
\langle \overline{\chi}_2(xp^{-2})|(-i{\rm {\bf k}})|\chi_2(1)\rangle
=\frac{(px^2; p^4)_\infty}{(x; p^4)_\infty}\langle 0 | 0 \rangle.
\end{align}
The linear forms are evaluated explicitly
by means of the standard formulas in $q$-analysis \cite{An}
like the $q$-binomial expansion and
\begin{equation*}
\sum_{j=0}^\infty\frac{(x; p)_j}{(p; p)_j}z^j =
\frac{(xz; p)_\infty}{(z; p)_\infty}.
\end{equation*}
The results are summarized in
\begin{proposition}\label{a:pr:exp}
For $j, m \in \Z_{\ge 0}$, the following formulas are valid:
\begin{equation}
\langle ({\rm {\bf a}}^\pm)^j{\rm {\bf k}}^m \rangle_{11}
= x^{\frac{1\pm 1}{2}j}p^{\frac{m}{2}+\frac{1\mp 1}{2}mj}
\frac{(-p; p)_j(x; p)_m}{(-px; p)_{j+m}}.
\end{equation}
\begin{equation}
\langle ({\rm {\bf a}}^\pm)^j{\rm {\bf k}}^m \rangle_{21}
= p^{\frac{m}{2}\mp jm}\frac{(x; p^2)_m(p;p)_j}{(-px; p^2)_{j+m}}
\sum_{i=0}^j(\mp 1)^ip^{\frac{i}{2}(i-(1\pm 1)j+1)}
\frac{(p^{2m}x; p^2)_i(-p^{2m+2i+1}x;p^2)_{j-i}}
{(p;p)_i(p;p)_{j-i}}.
\end{equation}
\begin{equation}
\begin{split}
&\langle ({\rm {\bf a}}^\pm)^{2j}{\rm {\bf k}}^{2m}\rangle_{22}
=x^{\frac{1\pm 1}{2}j}p^{m+(2\mp 2)mj}
\sum_{i=0}^j
\frac{(-1)^ip^{2i^2}(p^4;p^4)_j(x;p^4)_{m+i}}
{(p^4;p^4)_i(p^4;p^4)_{j-i}(xp^2;p^4)_{m+i}},\\
&\langle ({\rm {\bf a}}^\pm)^{2j}{\rm {\bf k}}^{2m+1}\rangle_{22}
=ix^{\frac{1\pm 1}{2}j}p^{m+(1\mp 1)(2m+1)j}
\sum_{i=0}^j
\frac{(-1)^ip^{2i^2}(p^4;p^4)_j(xp^2;p^4)_{m+i}}
{(p^4;p^4)_i(p^4;p^4)_{j-i}(xp^4;p^4)_{m+i}}.
\end{split}
\end{equation}
\end{proposition}
These formulas are useful for checks.
However, our proof of the main Theorem \ref{a:th:main}
does not rely on them.
We take $\boldsymbol{V}$ according to (\ref{a:choice})
and describe its decomposition
by introducing the base vectors as
\begin{align}
\boldsymbol{V}& =
\bigoplus_{{ \boldsymbol \alpha }=(\alpha_1,\ldots, \alpha_n) \in \{0,1\}^n}
 \C v_{ \boldsymbol \alpha },
 \label{a:V}\\
 \boldsymbol{V}& = \boldsymbol{V}_+ \oplus \boldsymbol{V}_- ,\qquad
 \boldsymbol{V}_\pm =  \bigoplus_{(-1)^{\alpha_1+\cdots + \alpha_n} = \pm 1}
 \C v_{ \boldsymbol \alpha }.
 \label{a:Vpm}
 \end{align}
Let $x$ be an indeterminate (spectral parameter).
The prescription (\ref{a:xLx})  in which the
$n$-layer $L$ operator (\ref{a:cL})  is built from the basic one
in  (\ref{a:Lmat}) leads to the following map
${\mathscr R}^{s, t}(x):
\boldsymbol{V} \otimes \boldsymbol{V}
\rightarrow  \boldsymbol{V} \otimes \boldsymbol{V}$
for $(s,t)=(1,1), (2,1)$ and $(2,2)$:
\begin{align}
&{\mathscr R}^{s, t}(x): \;\;
v_{ \boldsymbol \alpha }\otimes v_{\boldsymbol \beta} \mapsto
\sum_{{ \boldsymbol \alpha }', {\boldsymbol \beta}'\in \{0,1\}^n}
W_{st}\!\left(x\left|
{{ \boldsymbol \alpha }'  \;{\boldsymbol \beta}'
\atop \!{ \boldsymbol \alpha }\;\,{\boldsymbol \beta}}\right)\right.
v_{{\boldsymbol \alpha}'}\otimes v_{{ \boldsymbol \beta }'},
\label{a:Rmap}\\
&
W_{st}\!\left(x\left|
{{ \boldsymbol \alpha }'  \;{\boldsymbol \beta}'
\atop \!{ \boldsymbol \alpha }\;\,{\boldsymbol \beta}}\right)\right.
= \left\langle {\mathscr L} (\alpha'_1, \beta'_1|\alpha_1,\beta_1)
\cdots
{\mathscr L} (\alpha'_n, \beta'_n|\alpha_n,\beta_n)
\right\rangle_{st},
\label{a:wst}
\end{align}
where ${ \boldsymbol \alpha }=(\alpha_1,\ldots, \alpha_n)$,
${\boldsymbol \beta}=(\beta_1,\ldots, \beta_n)$, etc.
See Figure \ref{fig:W}.

\begin{figure}[ht]
\setlength{\unitlength}{0.20mm}
\begin{picture}(430,220)
\put(15,10){
\begin{picture}(400,200)
%
\path(40,20)(220,110)
\path(280,140)(360,180)
\dashline[0]{5}(220,110)(280,140)
\path(100,0)(100,100)\path(50,50)(150,50)
\path(200,60)(200,140)\path(160,100)(240,100)
\path(300,120)(300,180)\path(270,150)(330,150)
\put(-10,15){\scriptsize $\langle\bar{\chi}_s(x)|$}
\put(365,182){\scriptsize $|\chi_t(1)\rangle$}
\put(92,110){\scriptsize $\alpha_1^{}$}\put(92,-15){\scriptsize $\alpha_1'$}
\put(192,150){\scriptsize $\alpha_2^{}$}\put(192,45){\scriptsize $\alpha_2'$}
\put(292,190){\scriptsize $\alpha_n^{}$}\put(292,105){\scriptsize $\alpha_n'$}
\put(152,46){\scriptsize $\beta_1^{}$}\put(33,46){\scriptsize $\beta_1'$}
\put(242,96){\scriptsize $\beta_2^{}$}\put(143,96){\scriptsize $\beta_2'$}
\put(332,146){\scriptsize $\beta_n^{}$}\put(253,146){\scriptsize $\beta_n'$}
\end{picture}}
\end{picture}
\caption{A pictorial representation of
$W_{st}\!\left(x\left|
{{ \boldsymbol \alpha }'  \;{\boldsymbol \beta}'
\atop \!{ \boldsymbol \alpha }\;\,{\boldsymbol \beta}}\right)\right.$ in 
(\ref{a:wst}).}
\label{fig:W}
\end{figure}
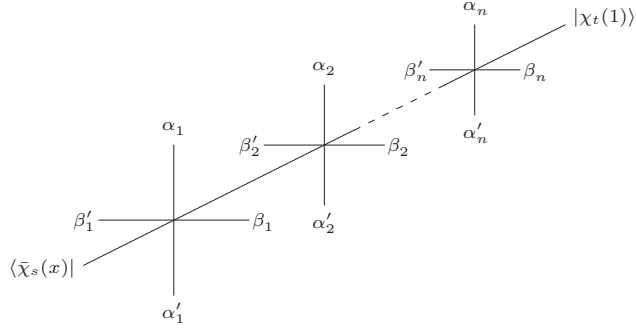

We remark that the construction (\ref{a:wst}) takes
a matrix product ansatz form for a spin chain whose local states
range over $V\otimes V$.
The matrix element \eqref{a:wst} depends on $x$ through
$\langle \overline{\chi}_s(x)| $ in the definitions \eqref{a:bracket} and
\eqref{a:bracket2}.
It is a rational function of $p^{\frac{1}{2}}$ and $x$
which is normalized to be $1$
for $({\boldsymbol \alpha }',  {\boldsymbol \beta}' )=
({\boldsymbol \alpha}, {\boldsymbol \beta})$ with
${\boldsymbol \alpha} = {\boldsymbol \beta}$.
For $(s,t)=(2,2)$, it also equals $1$
when
$({\boldsymbol \alpha }',  {\boldsymbol \beta}' )=
({\boldsymbol \alpha},  {\boldsymbol \beta})$
with
${\boldsymbol \alpha }-{\boldsymbol \beta}=(0,\ldots,0,\pm1)$
for which the corresponding product of ${\mathscr L} $'s in the bracket
is $-i{\rm {\bf k}} \in {\mathcal A}_{+-}$.

By the construction the decomposition
\begin{equation}
{\mathscr R}^{2,2}(x) =
{\mathscr R}^{2,2}_{+,+}(x)\oplus
{\mathscr R}^{2,2}_{+-}(x) \oplus
{\mathscr R}^{2,2}_{-+} (x)\oplus
{\mathscr R}^{2,2}_{--}(x)
\end{equation}
holds,
where ${\mathscr R}^{2,2}_{\varepsilon, \varepsilon'} (x):
 \boldsymbol{V}_{\varepsilon} \otimes  \boldsymbol{V}_{\varepsilon'}
\rightarrow
 \boldsymbol{V}_{\varepsilon} \otimes  \boldsymbol{V}_{\varepsilon'}$.
{}As explained in Section \ref{sec:teybe}, the $R$ matrix
${\mathscr R}(x)={\mathscr R}^{s, t}(x)$ satisfies the
Yang-Baxter equation (\ref{YBE}).
Another version
$\check{\mathscr R}^{s,t}(x)$ in (\ref{a:Rc})
is described as (\ref{a:Rmap}) by replacing
$v_{{\boldsymbol \alpha}'}\otimes v_{{ \boldsymbol \beta}'}$
with
$v_{{\boldsymbol \beta}'}\otimes v_{{ \boldsymbol \alpha}'}$
in the RHS.

\begin{example}
The image of the vector $v_{0,0,0}\otimes v_{1,1,0}$
is calculated as
\begin{equation*}
\begin{split}
{\mathscr R}^{s, t}(x)(v_{0,0,0}\otimes v_{1,1,0})=
&\langle ({\rm {\bf a}}^+)^2 \rangle_{st}v_{1,1,0}\otimes v_{0,0,0}+
\langle {\rm {\bf a}}^+(-i{\rm {\bf k}}) \rangle_{st}v_{1,0,0}\otimes v_{0,1,0}\\
+&\langle (-i{\rm {\bf k}}){\rm {\bf a}}^+ \rangle_{st}v_{0,1,0}\otimes v_{1,0,0}+
\langle (-i{\rm {\bf k}})^2 \rangle_{st}v_{0,0,0}\otimes v_{1,1,0}.
\end{split}
\end{equation*}
The matrix elements
are evaluated by using Proposition \ref{a:pr:exp}.
We list the result in the table.
($\langle {\rm {\bf k}} \,{\rm {\bf a}}^+ \rangle_{st}$
is equal to $p \langle {\rm {\bf a}}^+ {\rm {\bf k}} \rangle_{st}$.)

\begin{table}[ht]
\begin{center}
\begin{tabular}{c|c|c|c}
\hfill
& $({\rm {\bf a}}^+)^2$
& ${\rm {\bf a}}^+{\rm {\bf k}}$
& ${\rm {\bf k}}^2$ \\
\hline
&&\vspace{-0.3cm}\\
$\langle \;\;\rangle_{21}$
& \LARGE{$\frac{x(1+p)(1-p+px+p^3x)}{(1+px)(1+p^3x)}$}
& \LARGE{$\frac{x(1-x)p^{\frac{3}{2}}(1+p)}{(1+px)(1+p^3x)}$}
& \LARGE{$\frac{(1-x)p(1-p^2x)}{(1+px)(1+p^3x)}$}
\\

\hline
&&\vspace{-0.3cm}\\
$\langle \;\;\rangle_{11}$
& \LARGE{$\frac{x^2(1+p)(1+p^2)}{(1+px)(1+p^2x)}$}
& \LARGE{$\frac{x(1-x)p^{\frac{1}{2}}(1+p)}{(1+px)(1+p^2x)}$}
& \LARGE{$\frac{(1-x)p(1-px)}{(1+px)(1+p^2x)}$}
\\

\hline
&&\vspace{-0.3cm}\\
$\langle \;\;\rangle_{22}$
& \LARGE{$\frac{x(1-p^2)}{1-p^2x}$}
& \Large{$0$}
& \LARGE{$\frac{(1-x)p}{1-p^2x}$}
\end{tabular}
\end{center}
\end{table}

\end{example}

\section{Quantum $R$ matrices for spin representations}\label{a:sec:qr}

Consider the quantum affine Kac-Moody algebras
$U_q(B^{(1)}_n),
U_q(D^{(2)}_{n+1})$ and $U_q(D^{(1)}_n)$
without the derivation operator \cite{D86,Ji}.
The Dynkin diagrams of
$B^{(1)}_n, D^{(1)}_n$ and $D^{(2)}_{n+1}$
\cite{Kac} are given in Figure \ref{fig:Dynkin}.

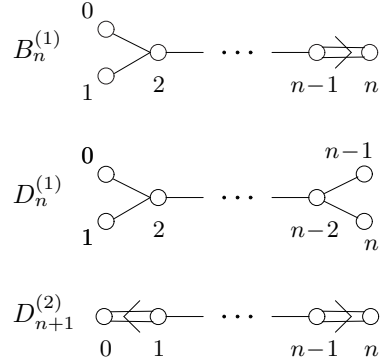
\begin{figure}

\begin{picture}(150,125)(-35,-100)
%
%
\put(0,10){
\put(0.3,8.8){\circle{6}}
\put(0.3,-9){\circle{6}}
\put(20,0){\circle{6}}
\put(80,0){\circle{6}}
\put(100,0){\circle{6}}
\put(45,0){\circle*{1}}
\put(50,0){\circle*{1}}
\put(55,0){\circle*{1}}
\drawline(3,7)(17,0)
\drawline(3.4,-8)(17,0)
\drawline(23,0)(37,0)
\drawline(63,0)(77,0)
\drawline(82,-2)(98,-1.8)
\drawline(82,2)(98,1.8)
\drawline(87,6)(93,0)
\drawline(87,-6)(93,0)
\put(-35,-2){$B^{(1)}_n$}
\put(-9,-18){\small $1$}
\put(-9,13){\small $0$}
\put(18,-15){\small $2$}
\put(70,-15){\small $n\!-\!1$}
\put(98,-15){\small $n$}
}
%
\put(0,-45){
\put(0.3,8.8){\circle{6}}
\put(0.3,-9){\circle{6}}
\put(20,0){\circle{6}}
\put(80,0){\circle{6}}
\put(98,9){\circle{6}}
\put(98,-9){\circle{6}}
\put(45,0){\circle*{1}}
\put(50,0){\circle*{1}}
\put(55,0){\circle*{1}}
\put(-9,-18){\small $1$}
\put(-9,13){\small $0$}
\drawline(3,7)(17,0)
\drawline(3.4,-8)(17,0)
\drawline(23,0)(37,0)
\drawline(63,0)(77,0)
\drawline(82.8,1.5)(95.5,7.8)
\drawline(82.8,-1.5)(95.5,-7.8)
\put(-35,-2){$D^{(1)}_n$}
\put(-9,-18){\small $1$}
\put(-9,13){\small $0$}
\put(18,-15){\small $2$}
\put(83,15){\small $n\!-\!1$}
\put(70,-15){\small $n\!-\!2$}
\put(98,-20){\small $n$}
}
%
\put(0,-90){
\put(-0.4,0){\circle{6}}
\put(20,0){\circle{6}}
\put(80,0){\circle{6}}
\put(100.4,0){\circle{6}}
\put(45,0){\circle*{1}}
\put(50,0){\circle*{1}}
\put(55,0){\circle*{1}}
\drawline(2,-1.8)(18,-2)
\drawline(2,1.8)(18,2)
\drawline(7,0)(13,-6)
\drawline(7,0)(13,6)
\drawline(23,0)(37,0)
\drawline(63,0)(77,0)
\drawline(82,-2)(98,-1.8)
\drawline(82,2)(98,1.8)
\drawline(93,0)(87,-6)
\drawline(93,0)(87,6)
\put(-35,-2){$D^{(2)}_{n+1}$}
\put(-2,-15){\small $0$}
\put(18,-15){\small $1$}
\put(70,-15){\small $n\!-\!1$}
\put(98,-15){\small $n$}
}

\end{picture}
\caption{The Dynkin diagrams for $B^{(1)}_n, D^{(1)}_n, D^{(2)}_{n+1}$
and their enumerations.}
\label{fig:Dynkin}
\end{figure}

Let $\{X^+_i, X^-_i, H_i |0\! \le \!i \le n\}$ be the Chevalley generators of
$U_q(B^{(1)}_n),
U_q(D^{(2)}_{n+1})$ and $U_q(D^{(1)}_n)$.
In $U_q(B^{(1)}_n)$ and $U_q(D^{(2)}_{n+1})$,
the classical part $\{X^+_i, X^-_i, H_i |1 \le i \le n\}$ forms the subalgebra
isomorphic to $U_q(B_n)$.
We endow $\boldsymbol{V}$ \eqref{a:V} with an irreducible action of $U_q(B_n)$.
For $U_q(D^{(1)}_{n})$
the classical part is of course $U_q(D_n)$.
Each of $\boldsymbol{V}_+$ and $\boldsymbol{V}_-$ individually admits
an irreducible action of $U_q(D_n)$.
By further supplementing these classical parts
with the ``0-action" of $X^\pm_0$ and $H_0$,
one gets the spin representations of
$U_q(B^{(1)}_n), U_q(D^{(1)}_{n})$ \cite{O} and
$U_q(D^{(2)}_{n+1})$.

Let us present the concrete formulas for them.
We realize $\boldsymbol{V}$ \eqref{a:V} as in  (\ref{a:choice}).
Thus we set $\boldsymbol{V} = V^{\otimes n}$ with
$V = \C v_0 \oplus \C v_1$ and identify the base
$v_{ \boldsymbol \alpha }$ with
${ \boldsymbol \alpha }=(\alpha_1,\ldots, \alpha_n) \in \{0,1\}^n$
with the tensor product as
\begin{equation}
v_{ \boldsymbol \alpha }=v_{\alpha_1}\otimes \cdots \otimes v_{\alpha_n}.
\end{equation}
Introduce the 2 by 2 matrices $X^+,X^-$ and $H$
acting on $V$ as
\begin{equation}\label{a:XH}
\begin{split}
&X^+v_0=0,  \quad \;X^-v_0 = v_1,  \quad H v_0=\textstyle\frac{1}{2}v_0,
\\
&X^+v_1=v_0,  \quad X^-v_1 = 0, \quad \,H v_1=-\textstyle\frac{1}{2}v_1.
\end{split}
\end{equation}
Then the action of the classical part is given as follows \cite{O}:
\begin{equation}\label{a:xpi}
\begin{split}
X^+_i &= - 1 \otimes \cdots \otimes 1\otimes \overset{i}{X^+} \otimes
\overset{i+1}{X^-} \otimes 1\otimes \cdots \otimes 1\\
H_i &= 1 \otimes \cdots \otimes 1\otimes(\overset{i}{H}\otimes \overset{i+1}{1}
-\overset{i}{1}\otimes \overset{i+1}{H})\otimes 1\otimes \cdots \otimes 1
\end{split}(1 \le i < n),
\end{equation}
which is common to all
the algebras $B^{(1)}_n, D^{(2)}_{n+1}$ and $D^{(1)}_n$.
On the other hand, $i=n$ case reads
\begin{align}
\begin{split}
X^+_n &= \frac{1}{\sqrt{q+q^{-1}}}1\otimes \cdots
\otimes 1 \otimes X^+\\
H_n &= 1 \otimes \cdots \otimes 1 \otimes H
\end{split}  \qquad  \text{for } B^{(1)}_n \text{ and } D^{(2)}_{n+1},
\label{a:xpn}\\
\begin{split}
X^+_n &= -1\otimes \cdots
\otimes 1 \otimes X^+\otimes X^+\\
H_n &= 1 \otimes \cdots \otimes 1\otimes({H}\otimes {1}
+{1}\otimes {H})
\end{split}\qquad  \text{for } D^{(1)}_n.
\label{a:xpn2}
\end{align}
For the algebras $B^{(1)}_n$ and $D^{(1)}_n$, the 0-action is given by
\begin{align}
X^+_0 &= -X^-\otimes X^- \otimes 1 \otimes \cdots \otimes 1,
\label{a:x0b}\\
H_0 &= -(H\otimes 1+1\otimes H)\otimes 1 \otimes \cdots \otimes 1.
\label{a:h0b}
\end{align}
For $D^{(2)}_{n+1}$, it takes the form
\begin{align}
X^+_0 &= -\frac{1}{\sqrt{q+q^{-1}}}X^-\otimes 1 \otimes \cdots \otimes 1,
\label{a:x0td}\\
H_0 &= -H\otimes 1 \otimes \cdots \otimes 1.\label{a:h0td}
\end{align}
In any case $X^-_i$ is given by the transpose ${}^t(X^+_i)$.
We note that the base vector
$v_{ \boldsymbol \alpha }$  of  $\boldsymbol{V}$ \eqref{a:V}  in this paper is
identified with $e_\mu$  in \cite{O} as
$v_{0,0,1} = e_{\frac{1}{2},\frac{1}{2},-\frac{1}{2}}$ etc.

The quantum $R$ matrix for the spin representation is a linear map
$R(x): \boldsymbol{V}\otimes \boldsymbol{V} \rightarrow
\boldsymbol{V} \otimes \boldsymbol{V}$
for $B^{(1)}_n$ and $D^{(2)}_{n+1}$.
Similarly, it is the linear map
$R(x): \boldsymbol{V}_{\varepsilon} \otimes
\boldsymbol{V}_{\varepsilon'} \rightarrow
\boldsymbol{V}_{\varepsilon}  \otimes \boldsymbol{V}_{\varepsilon'} $
for $D^{(1)}_n$ for each pair of $\varepsilon, \varepsilon' \in \{+,-\}$.
The composition
\begin{equation}\label{a:RP}
\check{R}(x) = P\, R(x)\quad (P(u \otimes v) = v \otimes u),
\end{equation}
is also called $R$ matrix. (We use the both in the sequel.)
Up to an overall scalar,
it is characterized by \cite{Ji}
\begin{align}
&[\check{R}(x), \Delta(g)] = 0 \quad \text{for } g =X^\pm_i, H_i\;(1 \le i \le n), \label{a:r1}\\
&\check{R}(x)(q^{H_0}\otimes X^+_0 + x X^+_0\otimes q^{-H_0})
=  (x q^{H_0}\otimes X^+_0 + X^+_0\otimes q^{-H_0})
\check{R}(x),\label{a:r2}
\end{align}
where the coproduct is specified as
\begin{equation}\label{a:Delta}
\Delta(X^{\pm}_i) = q^{H_i}\otimes X^{\pm}_i + X^\pm\otimes q^{-H_i},\quad
\Delta(H_i) = H_i\otimes 1+ 1 \otimes H_i.
\end{equation}
The quantum $R$ matrix $\check{R}(x)$
satisfies the same Yang-Baxter equation \eqref{a:ybe}
as $\check{\mathscr R}(x)$.

\begin{remark}
Our $\check{R}(x)$ is denoted by $R(x^{-1})$ in \cite{O}, which
can be seen by comparing  \eqref{a:r2} here and \cite[eq.(5.3)]{O}.
It is equal to ${\check R}(x^{-1})$ in the notation of \cite{Ji}.
\end{remark}

As mentioned before,
there are four kinds of $R$ matrices
$R(x): \boldsymbol{V}_{\varepsilon} \otimes
\boldsymbol{V}_{\varepsilon'} \rightarrow
\boldsymbol{V}_{\varepsilon}  \otimes \boldsymbol{V}_{\varepsilon'} $
for $D^{(1)}_n$.
We gather them into a single one
($2^{2n}$ by $2^{2n}$ matrix)
that acts on
$\boldsymbol{V} \otimes \boldsymbol{V}$
via \eqref{a:Vpm}, where the components other than
$\boldsymbol{V}_{\varepsilon} \otimes \boldsymbol{V}_{\varepsilon'} \rightarrow
\boldsymbol{V}_{\varepsilon}  \otimes \boldsymbol{V}_{\varepsilon'} $ are
to be understood as $0$.

We fix the normalization of
$R(x)$ by specifying a
particular matrix element as
\begin{equation}
\begin{split}
R(x)&:\; v_{0,\ldots,0}\otimes v_{0,\ldots,0} \mapsto
v_{0,\ldots,0}\otimes v_{0,\ldots,0} + \text{other terms}
\quad (B^{(1)}_n, D^{(2)}_{n+1}),\\
R(x)&:\; v_{0,\ldots,0,\alpha}\otimes v_{0,\ldots,0,\beta} \mapsto
v_{0,\ldots,0,\alpha}\otimes v_{0,\ldots,0,\beta} + \text{other terms}
\quad (D^{(1)}_{n}),
\end{split}
\end{equation}
where $\alpha, \beta \in \{0,1\}$ are arbitrary.
The resulting quantum $R$ matrices will be denoted by
$R_{B^{(1)}_{n}}(x),
R_{D^{(2)}_{n+1}}(x), R_{D^{(1)}_{n}}(x)$.
The ones obtained by them from (\ref{a:RP})
are similarly written as $\check{R}_{B^{(1)}_{n}}(x),
\check{R}_{D^{(2)}_{n+1}}(x),
\check{R}_{D^{(1)}_{n}}(x)$.
They are rational functions of $q$ and $x$,
and admits the spectral decomposition
$\check{R}(x)=\sum_{j=0}^n \rho^{(n)}_j(x)P^{(n)}_j$.
To explain $P^{(n)}_j$, recall the irreducible decomposition
of $U_q(B_n)$-module
\begin{equation*}
\boldsymbol{V} \otimes \boldsymbol{V} =
\boldsymbol{V}(2\Lambda_n) \oplus \boldsymbol{V}(\Lambda_{n-1})
\oplus \cdots \oplus \boldsymbol{V}(\Lambda_1) \oplus
\boldsymbol{V}(0),
\end{equation*}
where $\Lambda_j$ is the fundamental weight
attached to the vertex $j$ in the Dynkin diagram (Figure \ref{fig:Dynkin}),
and $\boldsymbol{V}(\lambda)$ denotes the irreducible
$U_q(B_n)$-module with highest weight $\lambda$.
(In this notation, the spin representation $\boldsymbol{V}$ on the LHS is
$\boldsymbol{V}(\Lambda_n)$. )
The operator $P^{(n)}_j$ in the spectral
decomposition is the orthnormal projector
from $\boldsymbol{V} \otimes \boldsymbol{V}$ to
$\boldsymbol{V} ((1\!+\!\delta_{j0})\Lambda_{n-j})$
where we set $\Lambda_0=0$.
For $B^{(1)}_n$, $P^{(n)}_j$ is described in \cite[Prop.5.1]{O},
which is actually the same also for  $D^{(2)}_{n+1}$ since
the two affine Lie algebras share the common classical part $B_n$.
See Figure \ref{fig:Dynkin}.
Projectors for $D^{(1)}_n$ are also associated with
similar irreducible decompositions of the $U_q(D_n)$-modules
$\boldsymbol{V}_{\varepsilon} \otimes
\boldsymbol{V}_{\varepsilon'}$.
See \cite{O} for the detail.

The eigenvalues for $B^{(1)}_n$ and $D^{(1)}_n$ are
given by
$\rho^{(n)}_j (x)= \tilde{\rho}^{(n)}_j(x^{-1})/\tilde{\rho}^{(n)}_{j_0}(x^{-1})$
where  $\tilde{\rho}^{(n)}_j(x)$ is the one given in \cite[p480, p482]{O}
with $j_0 = 0$ for $B^{(1)}_n$ and
$j_0 = (1-(-1)^j)/2$ for $D^{(1)}_n$.

The eigenvalues for $D^{(2)}_{n+1}$ seem nowhere available in the literature.
Although they are not necessary for the proof of our main Theorem \ref{a:th:main},
we present them for reader's convenience.
They are certainly vital for checks.
\begin{equation}\label{a:eigen}
\rho^{(n)}_j(x) = \prod_{i=1}^j\frac{q^{2i}+(-1)^ix}{q^{2i}x+(-1)^i}.
\end{equation}
We note that there are useful recursion formulas for the
$R$ matrices with respect to the rank $n$.
The one that matches our convention
for $B^{(1)}_n$ is obtained by setting
$x\rightarrow x^{-1}$ in \cite[Appendix]{O}
and for $D^{(1)}_n$  by setting in
$q \rightarrow q^2$ in \cite[sec. 2.5]{Koga}.
See also \cite{Re88} for $q=1$ case.

\section{Main theorem}\label{a:sec:proof}
To relate ${\mathscr R}(x)$
obtained from the 3d $L$ operators (Section \ref{a:sec:r3d}) and
$R(x)$ originating in the quantum affine algebras  (Section \ref{a:sec:qr}),
we adjust the parameter $p$ in the oscillator algebra ${\mathcal A}$
and $q$ in the quantum group $U_q$ by
\begin{equation}\label{a:pq}
p^{\frac{1}{2}} = i q.
\end{equation}
Now we state the main result of the paper.
\begin{theorem}\label{a:th:main}
With the identification \eqref{a:pq}, the following equalities are valid:
\begin{align}
&{\mathscr R}^{2,1}(x) = R_{B^{(1)}_n}(x),
\label{a:rb}\\
&{\mathscr R}^{1,1}(x) = R_{D^{(2)}_{n+1}}(x),
\label{a:rd2}\\
&{\mathscr R}^{2,2}(x) = R_{D^{(1)}_{n}}(x).
\label{a:rd}
\end{align}
\end{theorem}
\begin{remark}\label{re:dynkin}
Comparison of these results with Figure \ref{fig:Dynkin}
suggests the following correspondence between the boundary states
$\langle \overline{\chi}_s(x)|$, $|\chi_t(1)\rangle$ in (\ref{a:xLx})
and the end shape of the Dynkin diagrams:

\begin{picture}(200,100)(-60,-23)

\put(100,50){
\put(-0.4,0){\circle{6}}
\drawline(2,-1.8)(18,-2)
\drawline(2,1.8)(18,2)
\drawline(7,0)(13,-6)
\drawline(7,0)(13,6)
\put(-3,10){\small $0$}
\put(-50,-2){$\langle \overline{\chi}_1(x) |$}
}

\put(100,5){
\put(0.3,8.8){\circle{6}}
\put(0.3,-8.8){\circle{6}}
\drawline(3,7)(17,0)
\drawline(3,-7)(17,0)
\put(-9,-18){\small $1$}
\put(-9,13){\small $0$}
\put(-50,-2){$\langle \overline{\chi}_2(x) |$}
}

\put(90,50){
\drawline(82,-2)(98,-1.8)
\drawline(82,2)(98,1.8)
\drawline(93,0)(87,-6)
\drawline(93,0)(87,6)
\put(100.4,0){\circle{6}}
\put(120,-2){$|\chi_1(1)\rangle$}
\put(98,10){\small $n$}
}

\put(90,5){
\put(98,9){\circle{6}}
\put(98,-9){\circle{6}}

\drawline(82.8,0)(95.5,7.8)
\drawline(82.8,0)(95.5,-7.8)

\put(87,17){\small $n\!-\!1$}
\put(96,-20){\small $n$}
\put(120,-2){$|\chi_2(1)\rangle$}
}

\end{picture}

\noindent
In view of this, we expect that the similarly constructible
${\mathscr R}^{1,2}(x)$ yields the quantum $R$ matrix for
$U_q(B^{(1)}_n)$ corresponding to the realization of
$B^{(1)}_n$ as an affinization of its another classical subalgebra $D_n$.
We remark further that the periodic boundary condition along the
$F$-direction in \cite{BS06} corresponds to the
cyclic Dynkin diagram of the relevant algebra $A^{(1)}_n$
in an analogous way to the above pictures.
\end{remark}
\begin{proof}
In view of the normalizations,
it suffices to show that $\check{\mathscr R}^{s, t}(x)$ (\ref{a:Rc})
satisfies the characterization \eqref{a:r1} and \eqref{a:r2}.
For $g=H_i \,(1\le i \le n)$, eq.~\eqref{a:r1} is checked easily.
Thus our first task is to
show \eqref{a:r1}  for  $g=X^+_i$.
The case $g=X^-_i$ is similar and left as an exercise for the readers.

\smallskip
$\bullet$ {\it Proof of \eqref{a:r1} for $g=X^+_i$ with $1 \le i < n$}.
This case is most generic and relevant to all the algebras
$B^{(1)}_n, D^{(2)}_{n+1}$ and $D^{(1)}_{n}$.
It reads
\begin{equation}\label{a:com1}
\check{\mathscr R}^{s, t}(x)
(q^{H_i}\otimes X^+_i + X^+_i\otimes q^{-H_i})
=  (q^{H_i}\otimes X^+_i + X^+_i\otimes q^{-H_i})
\check{\mathscr R}^{s, t}(x).
\end{equation}
Our proof closes within the algebra ${\mathcal A}$ and is independent of its
representation. It neither concerns the choice of bra and ket vectors in
\eqref{a:bracket} and \eqref{a:bracket2}.
Therefore it applies to all of
$B^{(1)}_n, D^{(2)}_{n+1}$ and $D^{(1)}_{n}$.
To illustrate the idea, we consider the matrix element of
$\check{\mathscr R}^{s, t}(x)(q^{H_i}\otimes X^+_i)$
concerning the transition
$v_{{ \boldsymbol \alpha }}\otimes v_{{\boldsymbol \beta}}
\mapsto v_{{\boldsymbol \beta}'}\otimes v_{{ \boldsymbol \alpha }'}$.
By the action of $q^{H_i}\otimes X^+_i$, the vector
$v_{{ \boldsymbol \alpha }}\otimes v_{{\boldsymbol \beta}} $
firstly becomes
\begin{align}
&(-1)\times(\cdots \otimes
q^Hv_{\alpha_i}\otimes q^{-H}v_{\alpha_{i+1}}\otimes \cdots)
\otimes
(\cdots \otimes X^+v_{\beta_i}\otimes X^-v_{\beta_{i+1}}\otimes \cdots),
\end{align}
where the parts denoted by $\cdots$ are unchanged.
See \eqref{a:xpi}.
Concretely this represents the following vector in
$\boldsymbol{V}\otimes \boldsymbol{V}$:
\begin{equation}
-\delta_{\beta_i 1}\delta_{\beta_{i+1} 0}
q^{(\frac{1}{2}-\alpha_{i})-(\frac{1}{2}-\alpha_{i+1})}
(\cdots \otimes v_{\alpha_i}\otimes v_{\alpha_{i+1}}\otimes \cdots)
\otimes
(\cdots \otimes v_{0}\otimes v_{1}\otimes \cdots).
\end{equation}
See \eqref{a:XH}.
After further applying $\check{\mathscr R}^{s, t}(x)$ to this,
the coefficient of $v_{{\boldsymbol \beta}'}\otimes v_{{ \boldsymbol \alpha }'}$
in the resulting vector is $0$
unless $(\alpha'_j,\beta'_j)  = (\alpha_j, \beta_j)$ for all $j\neq i,i+1$.
If this condition is met,
the matrix element under consideration takes the form
\begin{equation}
\langle M
L(\alpha'_i, \beta'_i| q^H\alpha_i, X^+\beta_i)
L(\alpha'_{i+1}, \beta'_{i+1}| q^{-H}\alpha_{i+1}, X^-\beta_{i+1})N
\rangle_{st}
\end{equation}
for some elements $M, N\in {\mathcal A}$.
Here and in what follows, we employ the slightly abused notation like
\begin{equation}\label{a:abuse}
\begin{split}
&L(\alpha'_i, \beta'_i| q^H\alpha_i, X^+\beta_i)
= q^{\frac{1}{2}-\alpha_i}\delta_{\beta_i 1}
{\mathscr L} (\alpha'_i, \beta'_i| \alpha_i, 0),\\
&L(q^{-H}\alpha'_i, X^-\beta'_i| \alpha_i, \beta_i)
=q^{-\frac{1}{2}+\alpha'_i}\delta_{\beta'_i 0}
{\mathscr L} (\alpha'_i, 1| \alpha_i, \beta_i),
\end{split}
\end{equation}
where ${\mathscr L} (\alpha', \beta'| \alpha, \beta)$
is given by (\ref{a:Lmat}).
(Recall that the action of $\check{\mathscr R}^{s, t}(x)$ is
described by (\ref{a:Rmap}) and (\ref{a:wst}) followed by the transposition $P$
as in (\ref{a:Rc}).)
By similar calculations, one finds that the matrix element of
LHS-RHS of \eqref{a:com1}  concerning
$v_{{ \boldsymbol \alpha }}\otimes v_{{\boldsymbol \beta}}
\mapsto v_{{\boldsymbol \beta}'}\otimes v_{{ \boldsymbol \alpha }'}$
is proportional to
$\langle M Z_i N \rangle_{st}$ with
\begin{equation}\label{a:Lrel1}
\begin{split}
Z_i= &L(\alpha'_i, \beta'_i| q^H\alpha_i, X^+\beta_i)
L(\alpha'_{i+1}, \beta'_{i+1}| q^{-H}\alpha_{i+1}, X^-\beta_{i+1})\\
+
&L(\alpha'_i, \beta'_i| X^+\alpha_i, q^{-H}\beta_i)
L(\alpha'_{i+1}, \beta'_{i+1}| X^-\alpha_{i+1}, q^{H}\beta_{i+1})\\
-
&L(X^-\alpha'_i, q^H\beta'_i| \alpha_i, \beta_i)
L(X^+\alpha'_{i+1}, q^{-H}\beta'_{i+1}|\alpha_{i+1}, \beta_{i+1})\\
-
&L(q^{-H}\alpha'_i, X^-\beta'_i| \alpha_i, \beta_i)
L(q^H\alpha'_{i+1}, X^+\beta'_{i+1}|\alpha_{i+1}, \beta_{i+1})
\in {\mathcal A}.
\end{split}
\end{equation}
The four terms here correspond to those in \eqref{a:com1}.
Note for example in the composition
$(q^{H_i}\otimes X^+_i)\check{\mathscr R}^{s, t}(x)$,
one needs to look at the transition
$\check{\mathscr R}^{s, t}(x):
v_{ \boldsymbol \alpha }\otimes v_{\boldsymbol \beta}
\mapsto
v_{{\boldsymbol \beta}'}\otimes X^-_i v_{{ \boldsymbol \alpha }'}$
in order to finally reach the target vector
$v_{{\boldsymbol \beta}'}\otimes v_{{ \boldsymbol \alpha }'}$
by the subsequent action of
$q^{H_i}\otimes X^+_i$.
The element $Z_i$ can explicitly be written down
for each choice of $\alpha_i, \beta_i,\ldots,\alpha'_{i+1}, \beta'_{i+1}$
by substituting \eqref{a:Lmat} and \eqref{a:abuse}.
There are $2^8$ cases in total, and most of them are identically $0$.
By a direct calculation one can check that
all the nontrivial cases become 0 by using
the relation \eqref{a:kaa} and  \eqref{a:pq}.
In short, our claim is $Z_i = 0$, which is independent of the
representation of ${\mathcal A}$ and also of the
bra and ket vectors in
\eqref{a:bracket} and \eqref{a:bracket2} .

\smallskip
$\bullet$ {\it Proof of \eqref{a:r1} with $g=X^+_n$
for $B^{(1)}_n$ and $D^{(2)}_{n+1}$}.
The relevant $R$ matrices in \eqref{a:rb} and \eqref{a:rd2}
are ${\mathscr R}^{s, t}(x)$ with $t=1$.
Thus we are to show
\begin{equation}\label{a:com2}
\check{\mathscr R}^{s, 1}(x)(q^{H_n}\otimes X^+_n + X^+_n\otimes q^{-H_n})
=  (q^{H_n}\otimes X^+_n + X^+_n\otimes q^{-H_n})\check{\mathscr R}^{s, 1}(x),
\end{equation}
where $X^+_n$ and $H_n$ are specified in \eqref{a:xpn} .
By the same argument as before,
we are to check that
$\langle M Z_n \rangle_{s1}=0$ for any element $M \in {\mathcal A}$,
where $Z_n$ is given by
\begin{equation}\label{a:Lrel2}
\begin{split}
Z_n= &L(\alpha'_n, \beta'_n| q^H\alpha_n, X^+\beta_n)
+
L(\alpha'_n, \beta'_n| X^+\alpha_n, q^{-H}\beta_n)\\
-
&L(X^-\alpha'_n, q^H\beta'_n| \alpha_n \beta_n)
-
L(q^{-H}\alpha'_n, X^-\beta'_n| \alpha_n, \beta_n) \in {\mathcal A}.
\end{split}
\end{equation}
There are $2^4$ $Z_n$'s depending on the choices of
$\alpha_n,\ldots, \beta'_n$.
Writing them out one finds that
they all vanish by virtue of \eqref{a:kaa},
except the two nontrivial cases proportional to
${\rm {\bf a}}^\pm-1-iq^{\mp 1}{\rm {\bf k}}$.
Thus $\langle M Z_n \rangle_{s1}=0$ follows from
\eqref{a:aket1} and \eqref{a:pq}.

\smallskip
$\bullet$ {\it Proof of \eqref{a:r1} with $g=X^+_n$
for $D^{(1)}_n$}.
The relevant $R$ matrix in \eqref{a:rd} is
${\mathscr R}^{2,2}(x)$.
Thus we are to show
\begin{equation}\label{a:com3}
\check{\mathscr R}^{2,2}(x)(q^{H_n}\otimes X^+_n + X^+_n\otimes q^{-H_n})
=  (q^{H_n}\otimes X^+_n + X^+_n\otimes q^{-H_n})\check{\mathscr R}^{2,2}(x),
\end{equation}
where $X^+_n$ and $H_n$ are specified in \eqref{a:xpn2} .
As before we are to show
$\langle M Z'_n \rangle_{22}=0$ for any $M \in {\mathcal A}$, where
\begin{equation}\label{a:znp}
\begin{split}
Z'_n=
&L(\alpha'_{n-1}, \beta'_{n-1}| q^H\alpha_{n-1}, X^+\beta_{n-1})
L(\alpha'_{n}, \beta'_{n}| q^{H}\alpha_{n}, X^+\beta_{n})\\
+
&L(\alpha'_{n-1}, \beta'_{n-1}| X^+\alpha_{n-1}, q^{-H}\beta_{n-1})
L(\alpha'_{n}, \beta'_{n}| X^+\alpha_{n}, q^{-H}\beta_{n})\\
-
&L(X^-\alpha'_{n-1}, q^H\beta'_{n-1}| \alpha_{n-1}, \beta_{n-1})
L(X^-\alpha'_{n}, q^H\beta'_{n}|\alpha_{n}, \beta_{n})\\
-
&L(q^{-H}\alpha'_{n-1}, X^-\beta'_{n-1}| \alpha_{n-1}, \beta_{n-1})
L(q^{-H}\alpha'_{n}, X^-\beta'_{n}|\alpha_{n}, \beta_{n})
\in {\mathcal A}.
\end{split}
\end{equation}
There are $2^8$ $Z'_n$'s depending on $\alpha'_{n-1}, \ldots, \beta_n$.
Due to \eqref{a:kaa} and \eqref{a:pq}, they all vanish
except the two cases proportional to
$1-({\rm {\bf a}}^\pm)^2+q^{\mp 2}{\rm {\bf k}}^2$.
Thus $\langle M Z'_n \rangle_{22}=0$ follows from
\eqref{a:kaa}, \eqref{a:pq} and
the first relation in \eqref{a:aket2} with $x=1$.

\smallskip
The proof of \eqref{a:r1} has been finished.
Next we proceed to \eqref{a:r2}
concerning the 0-action.

\smallskip
$\bullet$ {\it Proof of \eqref{a:r2} for $B^{(1)}_n$ and $D^{(1)}_n$}.
We use \eqref{a:x0b} and \eqref{a:h0b}.
We are to check that
$\langle Z_0 N \rangle_{2t}=0$  for $t=1\, (B^{(1)}_n)$,
$t=2 \,(D^{(1)}_n)$ and
for any $N \in {\mathcal A}$, where
\begin{equation}\label{a:Lrel3}
\begin{split}
Z_0=
&L(\alpha'_1, \beta'_1| q^{-H}\alpha_1, X^-\beta_1)
L(\alpha'_2,\beta'_2|q^{-H}\alpha_2, X^-\beta_2)\\
+
&xL(\alpha'_1, \beta'_1| X^-\alpha_1, q^H\beta_1)
L(\alpha'_2,\beta'_2|X^-\alpha_2, q^H\beta_2)\\
-
&xL(X^+\alpha'_1, q^{-H}\beta'_1| \alpha_1, \beta_1)
L(X^+\alpha'_2, q^{-H}\beta'_2| \alpha_2, \beta_2)\\
-
&L(q^H\alpha'_1, X^+\beta'_1| \alpha_1, \beta_1)
L(q^H\alpha'_2, X^+\beta'_2| \alpha_2, \beta_2)
\in {\mathcal A}.
\end{split}
\end{equation}
There are $2^8$ $Z_0$'s depending on the choices of
$\alpha_1,\ldots, \beta'_2$.
Due to \eqref{a:kaa} and \eqref{a:pq},
they all vanish except the following:
\begin{equation}\label{a:z0}
{\rm {\bf a}}^+-x{\rm {\bf a}}^-,\quad {\rm {\bf a}}^{\pm}{\rm {\bf k}}
+(xq^2)^{\pm 1}{\rm {\bf k}}\,{\rm {\bf a}}^{\mp},
\quad
x^{\mp 1}({\rm {\bf a}}^\pm)^2-1- q^{\pm 2}{\rm {\bf k}}^2.
\end{equation}
The leftmost one makes zero contribution owing to the second relation in
\eqref{a:aket2}.
Therefore as far as the action on $\langle\overline{\chi}_2(x)|$ from the right
is concerned, one can write ${\rm {\bf a}}^+\equiv x {\rm {\bf a}}^-$.
Then the remaining ones in \eqref{a:z0}  are $\equiv$ to
${\rm {\bf a}}^{\pm}{\rm {\bf k}}+q^{\pm 2}{\rm {\bf k}}\,{\rm {\bf a}}^{\pm}$
and
${\rm {\bf a}}^\mp {\rm {\bf a}}^\pm-1- q^{\pm 2}{\rm {\bf k}}^2$.
These combinations are indeed $0$ because of
\eqref{a:kaa} and \eqref{a:pq}.
Note that this arguement holds either for $t=1$ or $t=2$
which concerns the choice of the ket vectors
in \eqref{a:bracket} and \eqref{a:bracket2}.

\smallskip
$\bullet$ {\it Proof of \eqref{a:r2} for $D^{(2)}_{n+1}$}.
We use \eqref{a:x0td} and \eqref{a:h0td}.
We are to check that
$\langle Z'_0 N\rangle_{11}=0$ for any $N \in {\mathcal A}$, where
\begin{equation}\label{a:Lrel4}
\begin{split}
Z'_0=
&L(\alpha'_1, \beta'_1| q^{-H}\alpha_1, X^-\beta_1)
+xL(\alpha'_1, \beta'_1| X^-\alpha_1, q^H\beta_1)\\
-&xL(X^+\alpha'_1, q^{-H}\beta'_1| \alpha_1, \beta_1)
-L(q^H\alpha'_1, X^+\beta'_1| \alpha_1, \beta_1)
\in {\mathcal A}.
\end{split}
\end{equation}
There are $2^4$ $Z'_0$'s depending on the choices of
$\alpha_1,\ldots, \beta'_1$.
Due to \eqref{a:kaa} and \eqref{a:pq},
they all vanish except
${\rm {\bf a}}^\pm - x^{\pm 1}(1+iq^{\pm 1}{\rm {\bf k}})$.
Thus $\langle \overline{\chi}_1(x)|Z'_0=0$ holds
thanks to \eqref{a:abra1}.

\medskip
Our proof of \eqref{a:r2} is finished and thereby
Theorem \ref{a:th:main} is established.
\end{proof}

\appendix
\section{Explicit formula of ${\mathscr R}$}\label{sec:app}
For reader's convenience, we quote from \cite{BMS:2008} the explicit formula for
${\mathscr R}={\mathscr R}_{1,2,3}$ in Theorem \ref{a:th:RLLL}
in a form free from implicit poles.
Define its matrix elements by
\begin{equation}
{\mathscr R}^{n'_1,n'_2,n'_3}_{n_1,n_2,n_3}
=\langle n_1, n_2, n_3 |
 {\mathscr R} |n'_1,n'_2, n'_3\rangle\quad (n_i, n'_i \in \Z_{\ge 0}),
\end{equation}
where $\langle n_1, n_2, n_3 | =\langle n_1| \otimes \langle n_2 | \otimes  \langle n_3 |$
and similarly for $ |n'_1,n'_2, n'_3\rangle$.
Then we have
\begin{align}
{\mathscr R}^{n'_1,n'_2,n'_3}_{n_1,n_2,n_3} &
= \sqrt{\frac{(p^2;p^2)_{n'_1}(p^2;p^2)_{n'_2}(p^2;p^2)_{n'_3}}
{(p^2;p^2)_{n_1}(p^2;p^2)_{n_2}(p^2;p^2)_{n_3}}}\;
\delta_{n_1+n_2, n'_1+n'_2}\delta_{n_2+n_3, n'_2+n'_3}\,
\overline{{\mathscr R}}^{n'_1,n'_2,n'_3}_{n_1,n_2,n_3},\\
\overline{{\mathscr R}}^{n'_1,n'_2,n'_3}_{n_1,n_2,n_3}
&=(-1)^{n'_2}\frac{p^{n_1n_3+n'_2(n_1+n_3+1)}}
{(p^2;p^2)_{n'_2}}\nonumber\\
&\times
\sum_{k=\max(0,n'_2-n_2)}^{\min(n_1,n_3,n'_2)}
\frac{(p^{-2n_1};p^2)_k(p^{-2n_3};p^2)_k(p^{-2n'_2};p^2)_k
(p^{2(n_2-n'_2+k+1)};p^2)_{n'_2-k}p^{2k}}
{(p^2;p^2)_k}.
\end{align}
This formula is equivalent to
\cite[eq.~(59)]{BMS:2008} without
the sign $(-1)^{n_2}$.
Removing it corresponds to setting
$\varepsilon =-1$ in \cite[eq.~(22)]{BMS:2008},
which is necessary 
to match eq. ~(\ref{themap}) in this paper.
The matrix elements enjoy the following symmetry:
\begin{equation*}
{\mathscr R}^{n'_1, n'_2, n'_3}_{n_1, n_2, n_3}  \,=\,
{\mathscr R}^{n_1, n_2, n_3}_{n'_1, n'_2, n'_3}  \,=\,
{\mathscr R}^{n'_3, n'_2, n'_1}_{n_3, n_2, n_1} .
\end{equation*}

\section*{Acknowledgements}
The authors thank
Vladimir Bazhanov and Vladimir Mangazeev for kind interest and 
encouragement.
A.K. thanks Masato Okado and Nobuyuki Furuyama
for communications. He also thanks
warm hospitality at Department of
Theoretical Physics, Australian National University
where a part of this work was done.
This work is supported by Grants-in-Aid for
Scientific Research No.~21540209 from JSPS
and partially by ARC.

\end{document}